\documentclass[10pt,conference,letterpaper]{IEEEtran}
\usepackage{times,amsmath,epsfig}
\usepackage{graphicx}
\usepackage{balance}

\usepackage{amsmath}
\usepackage{amssymb}
\usepackage{color}

\usepackage{subfigure}

\usepackage{graphicx}
\usepackage{csquotes}
\usepackage{url}
\usepackage{epstopdf}
\usepackage[ruled,vlined,boxed,linesnumbered]{algorithm2e}
\usepackage{algorithmic}

\newtheorem{theorem}{\hspace{-12pt} Theorem}[section]

\newtheorem{lemma}{ \hspace{-12pt}Lemma}[section]

\newtheorem{definition}{\hspace{-12pt}Definition}[section]

\newtheorem{example}{\hspace{-12pt}Example}[section]

\usepackage{multirow}

\usepackage[normalem]{ulem}
\usepackage{soul}
\usepackage{xcolor}

\IEEEoverridecommandlockouts

\title{A Generic Inverted Index Framework for Similarity Search on the GPU \\ \LARGE{[Technical Report]}}
\author{ Jingbo Zhou$^{\dagger*}$,
Qi Guo$^{\dagger}$, H. V. Jagadish$^{\#}$, Lubo\v{s} Kr\v{c}\'{a}l$^{\dagger}$, Siyuan Liu$^{\ddag}$\thanks{$^{\ddag}$Siyuan Liu has done his work on the project as an intern at NUS.}
\\ Wenhao Luan$^{\dagger}$,  Anthony K. H. Tung$^{\dagger}$, Yueji Yang$^{\dagger\S}$,
Yuxin Zheng$^{\dagger}$$^\S$\\
$~^{\dagger}$National University of Singapore
$~^{\#}$Univ. of Michigan, Ann Arbor
$~^{\ddag}$Nanyang Technological University\\
$~^{*}$Business Intelligence Lab, Baidu Research $~^{\S}$Tencent Inc.\\
$^{\dagger}$\{jzhou, qiguo, krcal, luan1, atung, yueji, yuxin\}@comp.nus.edu.sg\\
 $~^{\#}$jag@umich.edu, $~^{\S}$sliu019@e.ntu.edu.sg
}

\begin{document}
\maketitle

\begin{abstract}
We propose a novel generic inverted index framework on the GPU (called GENIE), aiming to reduce the programming complexity of the GPU for parallel similarity search of different data types.  Not every data type and similarity measure are supported by GENIE, but many popular ones are. We present the system design of GENIE, and demonstrate similarity search with GENIE on several data types along with a theoretical analysis of search results. A new concept of locality sensitive hashing (LSH) named $\tau$-ANN search, and a novel data structure c-PQ on the GPU are also proposed for achieving this purpose.  Extensive experiments on different real-life datasets demonstrate the efficiency and effectiveness of our framework. The implemented system has been released as open source\footnote{https://github.com/SeSaMe-NUS/genie}.
\end{abstract}

\vspace{-1mm}
\section{Introduction}
\vspace{-1mm}

There is often a need to support a high throughput of queries on index structures at scale.
These queries could arise from a multiplicity of ``users'', both humans and client applications.
Even a single application could sometimes issue a large number of queries.  For example, image matching is
often done by extracting hundreds of high dimensional
 SIFT (scale-invariant feature transform) features and matching them against SIFT features in the database.
Parallelization for similarity search is required for high performance on modern hardware architectures \cite{luo2012parallel,wang2013efficient,ding2009using,pan2011fast}. 

Solutions may be implemented it on Graphics Processing Units (GPUs). GPUs have experienced a tremendous growth in terms of computational power and memory capacity in recent years. One advanced GPU in the consumer market, the Nvidia GTX Titan X, has 12 GB of DDR5 memory
at a price of 1000 US dollars while an advanced server class GPU, the Nvidia K80, has
24GB of DDR5 memory at a price of 5000 US dollars. 
Furthermore, most PCs allow two to four GPUs to be installed, bringing the total amount of GPU memory in a PC to be comparable with a regular CPU memory.



However, GPU programming is not easy.  Effectively exploiting parallelism is even harder, particularly as we worry about the unique features of the GPU including the Single-Instruction-Multiple-Data (SIMD) architecture, concurrent control, coherent branching and coalescing memory access. While capable programmers could take their index structure of choice and create a GPU-based parallel implementation, doing so will require considerable effort and skill.

Our goal is to address this parallel similarity search problem on the GPU in a generic fashion. To this end, we develop an efficient and parallelizable GPU-based \underline{Gen}eric \underline{I}nverted Ind\underline{e}x framework, called GENIE  (we also name our system as GENIE), for similarity search using an abstract {\em match-count model} we define.
 GENIE is designed
to support parallel computation of the match-count model, but the system is generic in that a wide variety of data types and similarity measures can be supported. 

We do not claim that every data type and similarity measure are supported by GENIE\footnote{Note that we named our system as ``generic inverted index'', but not ``general inverted index''.} -- just that many are, as we will demonstrate in the paper, including most that we have come across in practice. As an analogy, consider the map-reduce model, implemented in a software package like Hadoop.  Not all computations can be expressed in the map-reduce model, but many can.  For those that can, Hadoop takes care of parallelism and scaling out, greatly reducing the programmer's burden. In a similar way, our system, GENIE, can absorb the burden of parallel GPU-based implementation of similarity search methods and index structures. 

Our proposed match-count model defines the common operations on the generic inverted index framework which has enough flexibility to be instantiated for different data types.
The insight for this possibility is that many data types can be transformed into a form that can be searched by an inverted-index-like structure. Such transformation can be done by the \emph{Locality Sensitive Hashing} (LSH) \cite{charikar2002similarity,datar2004locality} scheme under several similarity measures  or by the \emph{Shotgun and Assembly} (SA) \cite{aparicio2002whole,she2004shotgun} scheme for complex structured data.
We present a detailed discussion of this in Section \ref{sec:overview_lstsa}.



The first challenge of GENIE is to design an efficient index architecture for the match-count model.
We propose an inverted index structure on the GPU which can divide the query processing to many small tasks to work in a fine-grained manner to fully utilize GPU's parallel computation power. GENIE also exploits GPU's properties like coalescing memory access and coherence branching during the index scanning.

We propose a novel data structure on the GPU, called Count Priority Queue (c-PQ for short), which can significantly reduce the time cost for similarity search. Due to the SIMD architecture, another challenge of GENIE is how to select the top-k candidates from the candidate set, which is widely considered a main bottleneck for similarity search on the GPU in previous study \cite{alabi2012fast,pan2011fast} (which is called k-selection in \cite{alabi2012fast} and short-list search in \cite{pan2011fast}). Existing methods usually adopt a sorting method which is an expensive operation, or a priority queue on the GPU which have warp divergence problem and irregular memory movement.
Our novel design of c-PQ can keep only a few candidates on a hash table on the GPU, and we only need to scan the hash table once to obtain the query result. Therefore this major bottleneck for similarity search on the GPU can be overcome.


We optimize data placement on the GPU to improve the throughput of GENIE. The novel structure of c-PQ can also reduce the memory requirement for multiple queries. Therefore GENIE can substantially increase the number of queries within a batch on the GPU. We propose a tailored hash table on the GPU to reduce hash confliction. Besides, to overcome the limited memory size of the GPU, we introduce a multiple loading strategy to process large data.

We describe how to process data using an LSH scheme by GENIE. 
We propose a new concept, called Tolerance-Approximate Nearest Neighbour ($\tau$-ANN) search, which is in the same spirit as the popular $c$-ANN search.
Then we prove that, GENIE can support the $\tau$-ANN search for any similarity measure that has a generic LSH scheme.

For complex data types without LSH transformation, another choice is to adopt the SA scheme to process the data. We will
showcase this by performing similarity search on sequence data, short document data and relational data using GENIE.



We summarize our contributions as follows:
\begin{itemize}
\item We propose a generic inverted index framework (GENIE) on the GPU, which can absorb the burden of parallel GPU-based implementation of similarity search for any data type that can be expressed in the match-count model.
\item We present the system design of GENIE. Especially, we devise the novel data structure c-PQ to significantly increase the throughput for query processing on the GPU.
\item We exhibit an approach to adopting LSH scheme for similarity search under GENIE. We propose the new concept of $\tau$-ANN, and demonstrate that GENIE can effectively support $\tau$-ANN search under the LSH scheme.
\item We showcase the similarity search on complex data structures by GENIE under the SA scheme.
\item We conduct comprehensive experiments on different types of real-life datasets to demonstrate the effectiveness and efficiency of GENIE.
\end{itemize}

The paper is organized as follows. We will present an overview in Section \ref{sec:overview}, and expound system design in Section \ref{sec:index}, then we will discuss similarity search on GENIE in Section \ref{sec:search} and Section \ref{sec:search_org}. We will discuss experiment result in Section \ref{sec:exp} and related work in Section \ref{sec:related}. We will conclude the paper in Section \ref{sec:conclusion}.

\section{Preliminaries and Overview}\label{sec:overview}
In this section, we give an overview of GENIE including its main concepts and computational framework. We use relational data shown in Figure \ref{fig:invIndex} as a running example. 

\subsection{Match-count model}
Given a universe $U$, an \textbf{object} $O_i$ contains a set of elements in $U$, i.e. $O_i =\{o_{i,1},...,o_{i,r}\} \subset U$. A set of such data objects forms a \textbf{data set} $DS=\{O_1,...O_n\}$.
A \textbf{query} $Q_i$ is a set of items $\{q_{i,1},...,q_{i,s}\}$, where each item $q_{i,j}$ is a set of elements from $U$, i.e. $q_{i,j}\subset U$ ($q_{i,j}$ is a subset of $U$). A \textbf{query set} is defined as $QS=\{Q_1,...,Q_m\}$.

\begin{figure}[htb]
\vspace{-0mm}
\centerline{
\includegraphics[width=0.45\textwidth]{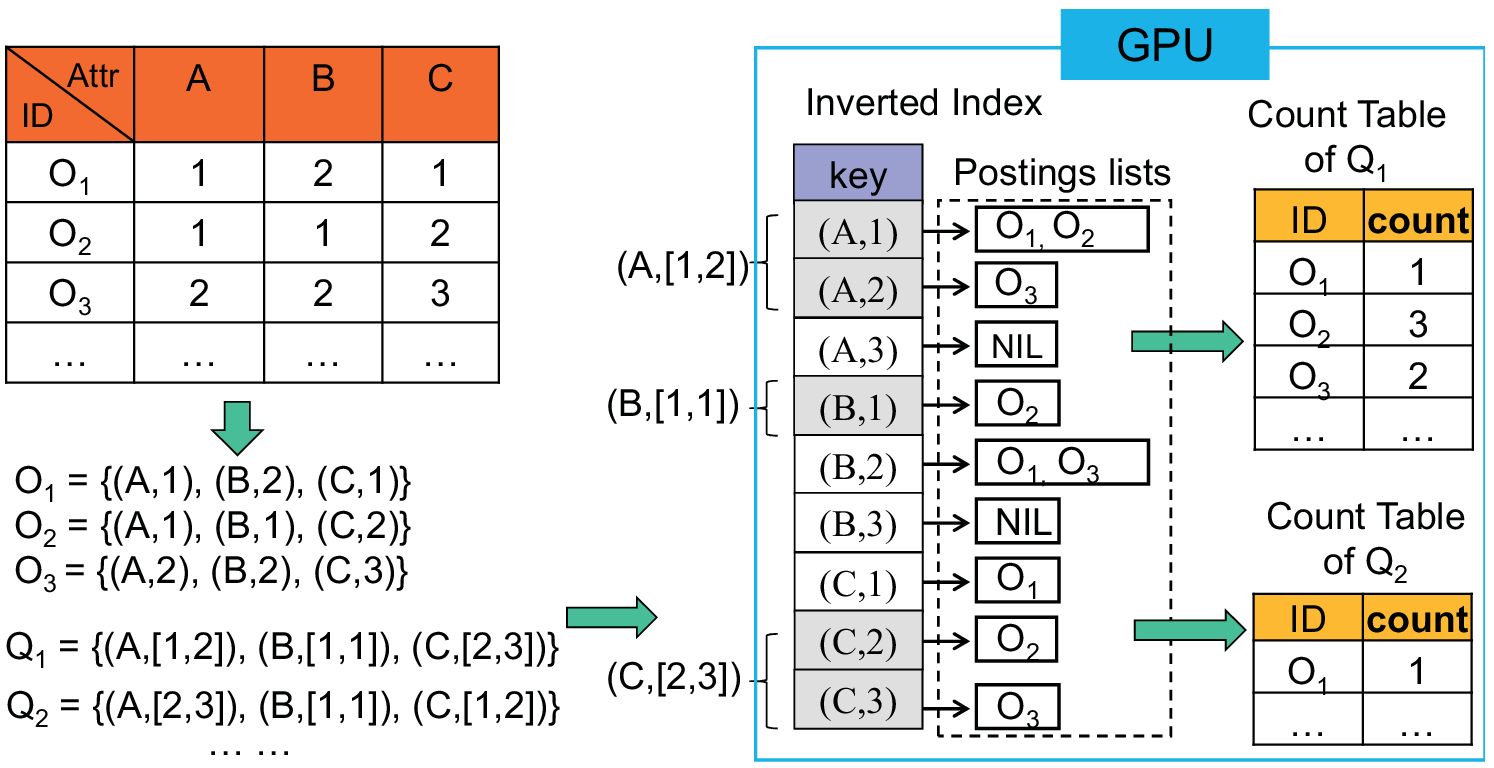}
}
\vspace{-0mm}
 \caption{An example on a relational table. 
 }
\label{fig:invIndex}
\vspace{-0mm}
\end{figure}


\begin{example}\label{exam:relation_tabel}
\emph{Given a relational table, the universe $U$ is a set of \emph{ordered pairs} $(d,v)$ where $d$ is an attribute of this table and $v$ is a value of this attribute. 
An $l-$dimensional relational tuple $p=(v_1,...,v_l)$ is represented as an object $O=\{(d_1,v_1),...,(d_l,v_l)\}$. As illustrated in Figure \ref{fig:invIndex}, the $O_1$ in the relational table is represented as $O_1=\{(A,1),(B,2),(C,1)\}$.}

\emph{A query on the relational table usually defines a set of ranges $R$$=$$([v^L_1,v^U_1]$,..., $[v^L_l,v^U_l])$.
Then it can be represented as $Q$=$\{r_1,r_2,..,r_l\}$, where $r_i=(d_i,[v^L_i,v^U_i])$ defines a set of pairs $(d_i,v)$ with value $v \in [v^L_i,v^U_i]$. As we can see from Figure \ref{fig:invIndex}, query $Q_1$ to retrieve the tuples with conditions $1\leq A\leq2, 1\leq B\leq1~and~2\leq C\leq3$ can be represented as $Q_1=\{(A,[1,2]),(B,[1,1]),(C,[2,3])\}$.}
\end{example}

Informally, given a query $Q$ and an object $O$, the match-count model $MC(\cdot,\cdot)$ returns the number of elements $o_{i}\in O$ contained by at least one query item of $Q$. We give a formal definition of the match-count model as follows.
\begin{definition}[match-count model]
\emph{Given a query $Q$ $=$ $\{r_1,r_2,..,r_l\}$ and an object $O=\{o_{1},...,o_{s}\}$, we map each query item $r_i$ to a natural integer using $C:(r_i,O) \rightarrow \mathbb{N}$, where $C(r_i,O)$ returns the number of elements $o_{j}\in O$ contained by the item $r_i$ (which is also a subset of $U$). Finally the output of the match-count model is the sum of the integers $MC(Q,O)=\sum_{r_i\in Q} C(r_i,O)$.
For example, in Figure \ref{fig:invIndex}, for $Q_1$ and $O_1$ we have $C((A,[1,2]),O_1)=1, C((B,[1,1]),O_1)=0~and~C((C,[2,3]),O_1)=0$, then we have $MC(Q_1,O_1)=1+0+0=1.$}
\end{definition}

In GENIE, we aim to rank all the objects in a data set with respect to the query $Q$ according to the model $MC(\cdot,\cdot)$ to obtain the top-$k$ objects of query $Q$.

GENIE essentially is an inverted index on the GPU to efficiently support the match-count model between objects and queries. The design of GENIE tries the best to utilize the GPU parallel computation capability to accelerate the query processing.  Figure \ref{fig:invIndex} shows an illustration of such high level inverted index. We first encode attributes and all possible values as ordered pairs (continuous valued attributes are first discretized).
Then we construct an inverted index where the \emph{keyword} is just the encoded pair and the \emph{postings list} comprises all objects having this keyword.
Given a query, we can quickly map each query item to the corresponding keywords (ordered pairs). After that, by scanning the postings lists, we can calculate the match counts between the query and all objects.

\subsection{GENIE with LSH and SA}\label{sec:overview_lstsa}

The inverted index with match-count model has the flexibility to support similarity search of many data types. 
Just as map-reduce model cannot handle all computation tasks, we do not expect that all data types can be supported. However, at least many popular data types can be supported with LSH or SA as we address further below. We illustrate the relationships among GENIE, LSH and SA in Fig. \ref{fig:lsh_genie_sa}. How to organize data structures as inverted indexes has been extensively investigated by previous literature \cite{wang2013efficient,Graph:Yan:SIGMOD:2005,yang2005similarity} and it is beyond the scope of this paper.

\subsubsection{Transformed by LSH.}
The most common data type, {high dimensional point}, can be transformed by an LSH scheme \cite{datar2004locality}. 
In such a scheme, multiple hash functions are used to hash data points into different buckets and points
that are frequently hashed to the same bucket are deemed to be similar. Hence we can build
an inverted index where each postings list corresponds to a list of points hashed to
a particular bucket. Given a query point, Approximate Nearest Neighbour (ANN) search can
be performed by first hashing the query point and then scanning
the corresponding postings list to retrieve data points that appear in many of these buckets.
Meanwhile, {sets, feature sketches and geometries} typically have kernelized similarity functions \cite{charikar2002similarity}, including Jaccard kernel for sets, 
Radial Basis Function (RBF) kernel for feature sketches, and Geodesic kernel for hyperplanes. 
We present the index building method under LSH scheme and theoretical analysis in Section \ref{sec:search}.
\vspace{-0mm}
\subsubsection{Transformed by SA}. 
The data with complex structure, including {documents, sequences, trees and graphs}, can be transformed with the SA \cite{aparicio2002whole,she2004shotgun} scheme. Specifically, the data will be broken down into smaller sub-units (``shotgun''), such as words for documents, 
n-grams for sequences \cite{wang2013efficient}, 
binary branches for trees \cite{yang2005similarity} and stars for graph  \cite{Graph:Yan:SIGMOD:2005}. After the decomposition, we can build an inverted index with a postings list for each unique sub-unit. Data objects containing a particular sub-unit are stored in the postings list. At query time, query objects will also be broken down into
a set of small sub-units and the corresponding postings lists will be accessed to find data objects that share  common sub-units with the query object. The match-count model returns the matching result between the query and objects sharing keywords, which is an important intermediate result for similarity search.
This approach has been widely used for similarity search of complex structured data \cite{wang2013efficient,yang2005similarity,Graph:Yan:SIGMOD:2005}. More discussion about it is presented in Section \ref{sec:search_org}.

\begin{figure}[htb]
\vspace{-0mm}
\centerline{
\includegraphics[width=0.35\textwidth]{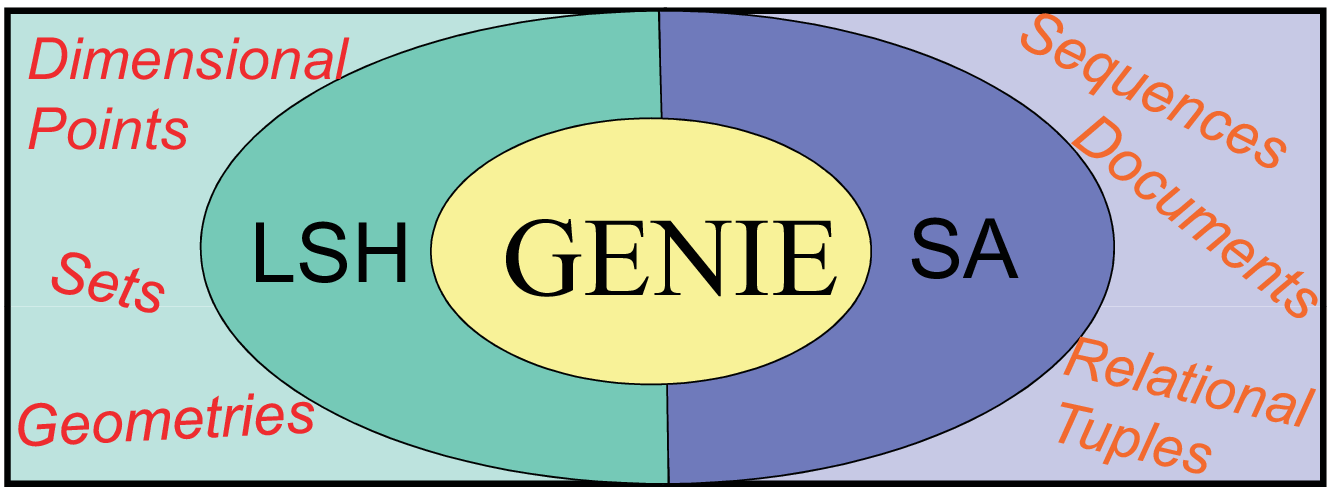}
}
\vspace{-0mm}
 \caption{  The relations among GENIE, LSH and SA.}
\label{fig:lsh_genie_sa}
\vspace{-0mm}
\end{figure}

\section{Inverted Index on the GPU}\label{sec:index}

We first give a brief introduction of the GPU, and then present the index structure and the data flow of GENIE. Next we present Count Priority Queue (c-PQ for short), which is a priority queue-like structure on the GPU memory facilitating the search. Finally, we propose a multiple loading method to handle large dataset. 

\subsection{Graphics Processing Unit}\label{sec:apx_gpu}
The Graphics Processing Unit (GPU) is a device that shares many aspects of Single-Instruction-Multiple-Data (SIMD) architecture. The GPU provides a massively parallel execution environment for many threads, with all of the threads running on multiple processing cores, and executing the same program on separate data. We implemented our system on an NVIDIA GPU using the Compute Unified Device Architecture (CUDA) toolkit \cite{cuda2015programming}. Each CUDA function is executed by an array of \emph{threads}. A small batch (e.g. 1024) of threads is organized as a \emph{block} that controls the cooperation among threads.

\subsection{Inverted index and query processing}\label{sec:index:gpuindex}

The inverted index is resident in the global memory of the GPU. Fig. \ref{fig:gpuIndex} illustrates an overview of such an index structure. All postings lists are stored in a large \emph{List Array} in the GPU's global memory. There is also a \emph{Position Map} in the CPU memory which stores starting and ending positions of each postings list for each keyword in the List Array. When processing queries, we use the Position Map to look up the corresponding postings list address for each keyword. This look-up operation is only required once for each query and our experiment also demonstrates that its time cost is negligible.

\begin{figure}[htb]
\vspace{-0mm}
\centerline{
\includegraphics[width=0.45\textwidth]{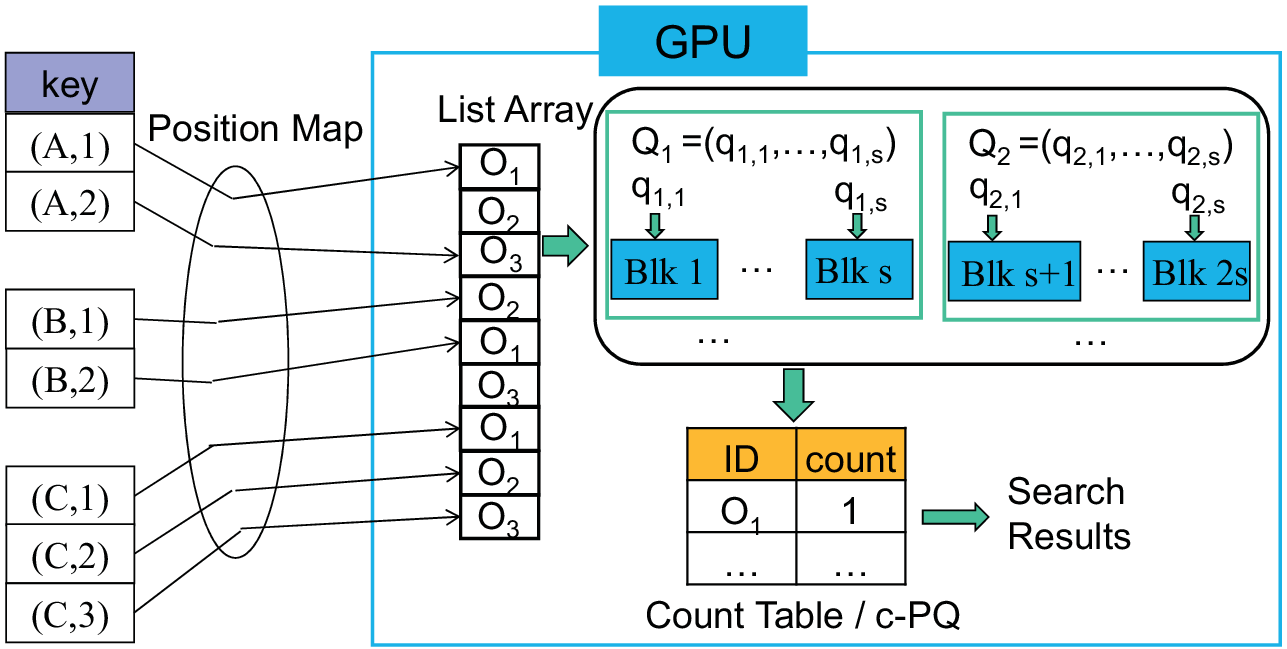}
}
\vspace{-0mm}
 \caption{  Overview of the inverted index and data flow.}
\label{fig:gpuIndex}
\vspace{-0mm}
\end{figure}

Figure \ref{fig:gpuIndex} shows the process of multiple queries on the GPU.  Each query has a set of items, which define particular ranges on some attributes. 
When we invoke a query, we first obtain its postings lists' addresses by the Position Map, then we use one block\footnote{A block on the GPU organizes a small batch of threads (up to 2048) and controls the cooperation among the threads.} of the GPU to scan the corresponding postings lists for each query item, where the threads of each block parallel access parts of the postings lists. For a query $Q_i=\{q_{i,1}, q_{i,2},...,q_{i,s}\}$ with $s$ query items, if there are $m$ queries, there will be about $m\cdot s$ blocks working on the GPU in parallel. During the process, after scanning an object in the postings list, each thread will update the \emph{Count Table} to update the number of occurrences of the object in the scanned postings lists. Therefore, the system works in a fine-grained manner to process multiple queries which fully utilizes the parallel computational capability of the GPU.

In the inverted index, there may be some extremely long postings lists, which can be bottleneck of our system. We also consider how to balance the workload for each block by breaking long postings lists into short sub-lists. 

\subsubsection{Load balance} \label{sec:load_balance}
\begin{figure}[htb]
\vspace{-0mm}
\centerline{
\includegraphics[width=0.4\textwidth]{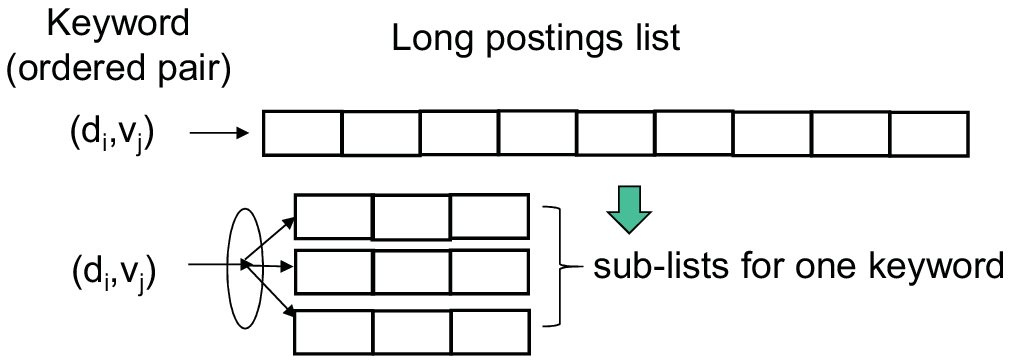}
}
\vspace{-0mm}
 \caption{Splitting long postings list for load balance.}
\label{fig:loadBalance}
\vspace{-0mm}
\end{figure}

For the inverted index, there may be some extreme long postings lists which can become the bottleneck of the system. Thus, it is necessary to consider the load balance problem in such application context. We also implement a load balance function, whose solution is to limit the length of postings lists. When the postings list is too long, we divide such a long postings list into a set of sublists, we build a one-to-many position map to store the addresses of the sub-lists in the List Array. Figure \ref{fig:loadBalance} gives an illustration for splitting a long postings list to three postings sub-lists. During scanning the (sub-)postings lists, we also limit the number of lists processed by one block. In our system, after enabling the load balance function, we limit the length of each (sub-)postings list as 4K and each block takes two (sub-)postings lists at most. It is worthwhile to note that, if there are already many queries running on the system, the usefulness of load balance is marginally decreased. This is because all the computing resources of the GPU have been utilized when there are many queries and the effect of load balance becomes neglected.

\subsection{Count Priority Queue}\label{sec:index:ch}
We propose a novel data structure, called Count Priority Queue (c-PQ for short) to replace the Count Table on the GPU, which aims to improve the efficiency and throughput of GENIE. c-PQ has two strong points: 1) Though how to retrieve the top-k result from all candidates is the major bottleneck for similarity search on the GPU \cite{pan2011fast}, c-PQ can finish this task with small cost; and 2) c-PQ can significantly reduce the space requirement of GENIE.



One major obstacle of GENIE is how to select top-$k$ count objects from the Count Table. This problem is also considered a major bottleneck for similarity search on the GPU in previous study \cite{pan2011fast}. It is desirable to use a priority queue for this problem. However, the parallel priority queue usually has warp divergence problem and irregular memory movement, which cannot run efficiently on GPU architectures \cite{he2012design}. 

The key idea of c-PQ is to use a two-level data structure to store the count results, and use a device to schedule the data allocation between levels. Our novel design can guarantee that only a few of candidates are stored in the upper level structure while all objects in lower level structure can be abandoned. Then we only need to scan the upper level structure to select the query result. We will describe the structure and mechanism of c-PQ, and prove all the claims in Theorem \ref{theorem:atk}.


Another major problem of GENIE is its large space cost, since the Count Table must allocate integer to store the count for each object for each query. Taking a dataset with 10M points as an example, if we want to submit a batch of one thousand queries, the required space of the Count Table is about 40 GB (by allocating one integer for count value, the size is $1k (queries) \times 10M (points) \times 4(bytes) =40 GB$), which exceeds the memory limit of the current available GPUs.  

To reduce space cost, first, we can use bitmap structure to avoid explicitly storing id.
Second, we only need to allocate several (instead of 32) bits to encode the count for each object in bitmap structure. The reason is that the maximum count for each object is bounded (i.e. there is a maximum value of the count) since the count value cannot be larger than the number of postings lists in the index. Actually, we usually can infer a much smaller count bound than the number of postings lists. For example, for high dimensional points, the maximum count value is just the number of its dimensions. 

\subsubsection{The structure and mechanism of c-PQ}\label{sec:sm_cpq}

\begin{figure}[htb]
\vspace{-0mm}
\centerline{
\includegraphics[width=0.45\textwidth]{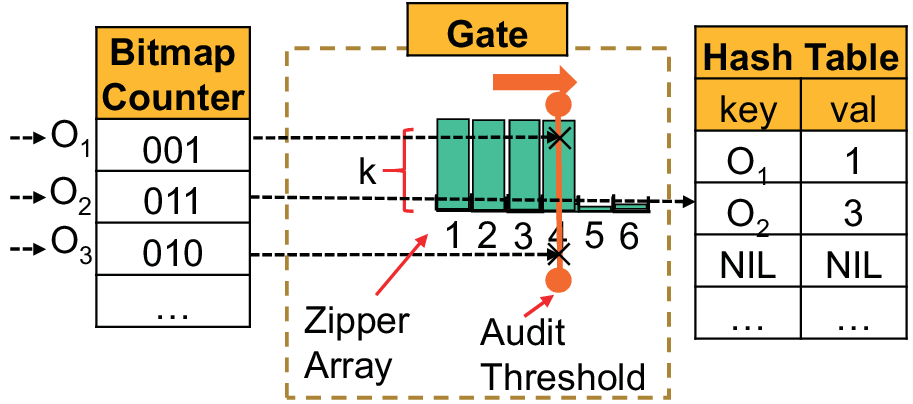}
}
\vspace{-0mm}
 \caption{  An illustration of the c-PQ.}
\label{fig:countHeap}
\vspace{-0mm}
\end{figure}

Figure \ref{fig:countHeap} shows the main structures of c-PQ which has three components. In the lower level, we create a \emph{Bitmap Counter} which allocates several bits (up to 32 bits) for each objects whose id is corresponding to the beginning address of the bits. In the upper level, there is a \emph{Hash Table} whose entry is a pair of object id and its count value. Then, a pivotal device, called \emph{Gate}, determines which id-value pair in the Bitmap Counter will be inserted into the Hash Table. The Gate has two members: a \emph{ZipperArray}  and a threshold called \emph{AuditThreshold}. In the following context, we briefly denote the Bitmap Counter as $BC$, the Hash Table as $HT$, the AuditThreshold as $AT$ and the ZipperArray as $ZA$.

The Gate has two functions. First, only a few objects in the BC can pass the Gate to the HT, while all objects remaining in the BC cannot be top $k$ objects and thus can be safely abandoned. Second, the AT in the Gate just keeps track of the threshold for the top-$k$ result, and we only need to scan the HT once to select the objects with counter larger than $AT-1$ as top-$k$ results (See Theorem \ref{theorem:atk}).

The ZA and AT in the Gate work together to restrict objects going from the BC to the HT. 
The size of the \emph{ZA} in Gate is equal to the maximum value of the count (based on the count value bound). \emph{ZA[i]}\footnote{ZA is 1-based indexing array, i.e. the index starts from 1.} records the minimum value between the number of objects whose count have reached $i$ (denoted as $zc_i$) and value $k$, i.e. $ZA[i]=min(zc_i,k)$. The \emph{AT} in Gate records the minimum index of ZA whose value is smaller than $k$ (i.e. $ZA[AT]<k$ and $ZA[AT-1]\geq k$).

The intuition behind the mechanism of the Gate is that, if there are already $k$ objects whose count has reached $i$ (i.e. $ZA[i]==k$) in the HT, there is no need to insert more objects whose count is less or equal to $i$ into the HT since there are already $k$ candidates if the top-$k$ count threshold is just $i$. Therefore, the $AT$ increase by 1 when $ZA[AT]==k$. 




We present the update process per thread on the GPU of c-PQ in Algorithm \ref{alg:countHeap}, which is also illustrated in Fig. \ref{fig:countHeap}. For each query, we use one block of the GPU to parallel process one query item. For all inverted lists matched by a query item, the threads of this block access the objects and update c-PQ with Algorithm \ref{alg:countHeap}. Note that the add operation is atomic.
When the count of an object is updated in the BC, we immediately check whether the object's count is larger than the \emph{AT} (line \ref{alg:countHeap:check}). If it is, we will insert (or update) an entry into the HT whose key is the object id and whose value is the object's count. 
Meanwhile, we will update the \emph{ZA} (line \ref{alg:countHeap:aa}). If \emph{ZA[AT]} is larger than $k$, we also increase the \emph{AT} by one unit (line \ref{alg:countHeap:zi}).

{\vspace{-0mm}
\begin{algorithm}[!htb]
\small
\DontPrintSemicolon
 \SetKwInOut{Input}{input}\SetKwInOut{Output}{output}

\tcp{For a thread in a block, it accesses object $O_i$ in the inverted index, then makes following updates.}

\ShowLn $val_i = BC[O_i]+1$

\ShowLn $BC[O_i] = val_i$

\If{$val_i\geq AT$}{\label{alg:countHeap:check}


        \ShowLn Put entry $(O_i,val_i)$ into the HT

        \ShowLn $ZA[val_i] += 1$ \label{alg:countHeap:aa}

        \While{ $ZA[AT] \geq k$}{\label{alg:countHeap:azk}
            \ShowLn $AT += 1$\label{alg:countHeap:zi}
        }
    }
 \caption{Update on the Count Priority Queue}\label{alg:countHeap}
 \vspace{-0mm}

\end{algorithm}
\vspace{-0mm}
}





For the Algorithm \ref{alg:countHeap} of c-PQ, we have the following lemma.
\begin{lemma}\label{lemma:aak}
\vspace{-0mm}
\emph{In Algorithm \ref{alg:countHeap}, after finishing all the updates of the Gate, we have $ZA[AT]<k$ and $ZA[AT-1]\geq k$.
}\vspace{-0mm}
\end{lemma}
\begin{proof}
\vspace{-0mm}
\emph{In  Algorithm \ref{alg:countHeap}, after each update of \emph{ZA} in line \ref{alg:countHeap:aa},  we check whether $ZA[AT] \geq k$ in line \ref{alg:countHeap:azk}. If it is, we increase $AT$ (in line \ref{alg:countHeap:zi}). Therefore, we always have   $ZA[AT]<k$. Similarly, since we only increase $AT$ in line \ref{alg:countHeap:zi}, we can guarantee that $ZA[AT-1]\geq k$.}
\vspace{-0mm}
\end{proof}

Now we use  Theorem \ref{theorem:atk} to elaborate the properties of c-PQ mentioned above. 
\begin{theorem}\label{theorem:atk}
\vspace{-0mm}
\emph{After finishing scanning the inverted index and updating c-PQ, the top-$k$ candidates are stored in the HT, and the number of objects in the HT is $O(k*AT)$.
Suppose the match count of the $k$-th object $O_k$ of a query $Q$ is $MC_k=MC(Q,O_k)$, then we have $MC_k = AT -1$.}
\vspace{-0mm}
\end{theorem}
\begin{proof}
\emph{For the claim that ``the top-$k$ candidates are only stored in the HT (not in the BC)'', since there at least $O(k)$ objects passing the gate for each possible value of AT, therefore, the objects left in BC cannot be the top-$k$ candidates. For the same reason, the objects in the HT is $O(k*AT)$.}

\emph{For ``$MC_k = AT -1$'', we prove it by contradiction. On the one hand, if we suppose $MC_k > AT -1$, we can deduce that  $MC_k \geq AT$. Thus, we can further infer that $ZA[AT]\geq k$, which contradicts with Lemma \ref{lemma:aak}. On the other hand, if we suppose $MC_k < AT -1$, then there must be less than $k$ objects with match count greater than or equal to $AT-1$, which contradicts with  Lemma \ref{lemma:aak}. Therefore, we only have $MC_k = AT -1$.}
\end{proof}


According to Theorem \ref{theorem:atk}, we can select the top-$k$ objects by scanning the HT and selecting objects with match count greater than ($AT -1$) only. If there are multiple objects with match count equal to ($AT -1$), we break ties randomly.

We give an example to show update process of c-PQ with data in Figure \ref{fig:invIndex} and the final result shown in Figure \ref{fig:countHeap}.
\begin{example}
\emph{Given a data set $\{O_1, O_2, O_3\}$ and a query $Q_1$ in Figure \ref{fig:invIndex}, we want to find the top-1 result of $Q_1$ from the objects, i.e. $k=1$. Since the number of attributes of the table is $3$, the maximum value of count is $3$.  Initially we have $AT=1$, $ZA=[0,0,0]$, $BC=\{O_1:0, O_2:0, O_3:0\}$ and $HT=\emptyset$.
For easy explanation, we assume the postings lists matched by $Q_1$ are scanned with the order of $(A, [1,2]), (B,[1,1])$ and $(C,[2,3])$. (On the GPU they are processed with multiple blocks in parallel with random order.)}

\emph{As shown in Algorithm \ref{alg:countHeap}, when scanning the postings list $(A, [1,2])$, we first access $O_1$ and get $BC(O_1)=1$. Since $BC(O_1) \geq AT(=1)$, we have $HT(O_1)=1$ and $ZA[1]=0+1=1$ ( note that ZA is 1-based array, thus after the updating $ZA=[1,0,0]$). Since $ZA[AT] \geq k$($k=1$ and $AT=1$), then we have $AT = 1+1 =2$.
 Then we update $BC(O_2)=1$ and $BC(O_3)=1$ without changing other parameters since both the values of $O_2$ and $O_3$ are smaller than $AT$.}

\emph{With the same method, after processing $(B,[1,1])$, we only have $BC(O_2)=2$ , $HT(O_2)=2$, $ZA=[1,1,0]$, $AT=3$ and $BC=\{O_1:1, O_2:2, O_3:1\}$.}

\emph{The last postings list to process is $(C,[2,3])$. There is no $O_1$ in it. For $O_2$, we have $BC(O_2)=3$. Since $BC(O_2) \geq AT$, we have $HT(O_2)=3$, $ZA=[1,1,1]$ and $AT=4$. We also have $BC(O_3)=2$.}

\emph{Finally, we have $HT=\{O_1:1, O_2:3\}$ and $AT=4$. By Theorem \ref{theorem:atk}, we know that the count of top-1 result is 3 (AT-1=4-1=3). We can then scan the hash table $HT$ to select the object equal to $3$ which is just $O_2$.}
\end{example}

\subsubsection{Hash Table with modified Robin Hood Scheme}

Here we briefly discuss the design of Hash Table in c-PQ. 
We propose a modified Robin Hood Scheme to implement a hash table on the GPU which is different from existing work \cite{garcia2011coherent}. According to Theorem \ref{theorem:atk}, the size of the Hash Table can be set as $O(k*max\_count\_value)$. We adopt a lock-free synchronization mechanism which is studied in \cite{moazeni2012lock} to handle the race condition problem. 

The general idea of the Robin Hood hashing is to record the number of probes for reaching (inserting or accessing) an entry due to conflicts.
We refer to this probing number as age.
During the insertion, one entry will evict an existing entry in the probed location if the existing entry has a smaller age.
Then we repeat the process of inserting the evicted entry into the hash table until all the entries are inserted.
For a non-full table of size $T$, the expected maximum age per insertion or access is $\Theta(log_2 log T)$ \cite{devroye2004worst}.

The vital insight to improve the efficiency of the Robin Hood Scheme in the c-PQ is that all entries with values smaller than ($AT -1$) in the Hash Table cannot be top-$k$ candidates (see Theorem \ref{theorem:atk}). 
If the value of an entry is smaller than ($AT -1$), we can directly overwrite the entry regardless of hashing confliction. Thus we can significantly reduce the probe times of insertion of the HT, since many inserted keys become expired with the increase of \emph{AT}.  



\subsection{Indexing large data with multiple loadings}\label{sec:multiple_loadings}

We also devise a multiple loading method to increase the capacity of GENIE utilizing the advantage of the high GPU memory bandwidth. 
There may be some cases that the data index is too large to be fitted into the GPU memory. For this problem, we split the whole data set into several parts, and then build inverted index for each part in the CPU memory. When a batch of queries is submitted, we transfer each index part into the GPU memory in turn, and run the query processing introduced before. After finishing a round, we collect all results from each part, and merge them to get a final query results. Some necessary computation is done in the CPU such as finding the final top-$k$ results among the top-$k$ results of each data part. Fig. \ref{fig:multiple_loadings} illustrates this multiple loading method.

\begin{figure}[htb]
\vspace{-0mm}
\centerline{
\includegraphics[width=0.45\textwidth]{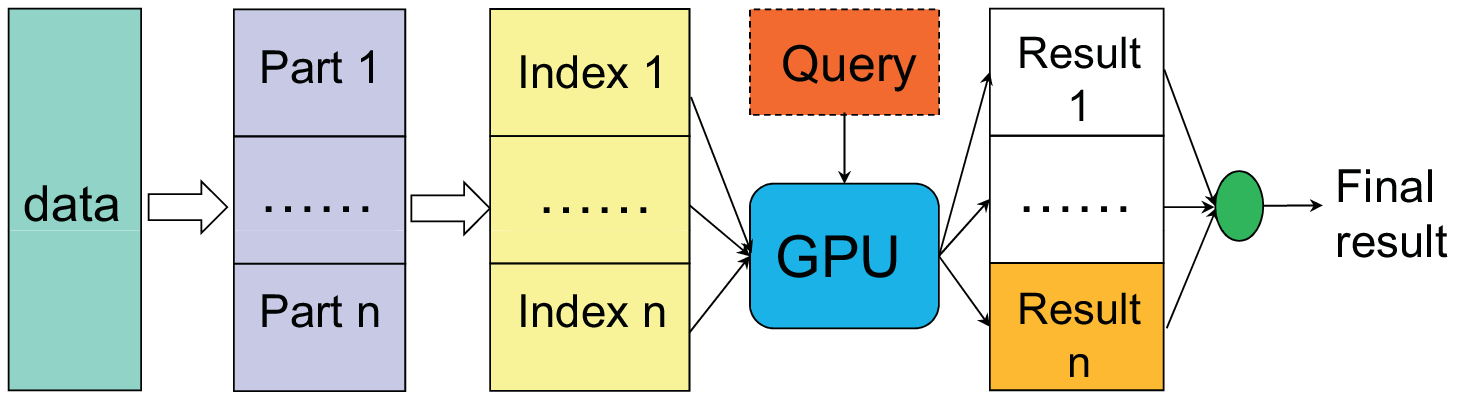}
}
\vspace{-0mm}
 \caption{  An illustration of GENIE with multiple loadings.}
\label{fig:multiple_loadings}
\vspace{-0mm}
\end{figure}

\subsection{The utility of the GPU for GENIE}\label{sec:multiple_loadings}
The design of GENIE utilizes the property of the SIMD architecture and the features of the GPU from several perspectives. First, GENIE divides the query process to sufficient number of small tasks to fully utilizes the parallel computation power of the GPU. We use one block to handle one query item and use each thread of the block to process an object in the postings list, therefore, the similarity search can be processed in as fine-grained manner as possible, so that all the processors on the GPU are highly utilized.  

Second, c-PQ can finish the top-k selection task with a small number of homogenous operations (scanning the Hash Table only once) suitable for the SIMD architecture. This data structure avoids the expensive operations like sort \cite{pan2011fast,alabi2012fast} and data dependent memory assess movement on the GPU \cite{he2012design} for top-k selection. This point is clarified in the experiment of Section \ref{exp:search_time} and Section \ref{sec:exp:cpq}, with a discussion about its competitors in Section \ref{exp:genie_discussion}.

Third, the design of GENIE tries to perform coalescing memory access and coherent branching.  In c-PQ, most operations are limited in the BC and only a few data are passed to the HT, which minimizes the branch divergence. Besides, since we use many threads to process a postings list, the threads have coalescing memory access patterns.

Fourth, the multiple loading method takes the advantage of the high GPU memory bandwidth which is usually 5-10 times higher than the CPU memory bandwidth.
Our experiment in Section \ref{sec:sift_multiple_loadings} also demonstrates that such index transfer step only takes a very small portion of the total time cost.

\section{Generic ANN Search with LSH }\label{sec:search}
Here we first show that GENIE can support the Approximate Nearest Neighbor (ANN) search for any similarity measure which has a generic LSH scheme, followed by an estimation of error bound for ANN search on GENIE.

\subsection{Building index for ANN search on GENIE}
In this section, we show how to use GENIE to support similarity search after processing data by LSH scheme.
\subsubsection{ANN search with LSH on GENIE}
According to the definition in \cite{charikar2002similarity}, a hashing function $h(\cdot)$ is said to be locality sensitive if it satisfies:
\begin{equation}\label{eqn:simlsh}
\vspace{-0mm}
Pr[h(p)=h(q)] = sim(p,q)
\vspace{-0mm}
\end{equation}
which means the collision probability is equal to the similarity measure. Here $sim(\cdot,\cdot)$ is a function that maps a pair of points to a number in $[0,1]$ where $sim(p,q)=1$ means $p$ and $q$ are identical. LSH is one of the most popular solutions for the ANN search problem \cite{indyk1998approximate,charikar2002similarity,wang2014hashing}. 


We can use the indexing method for relation table shown in Figure \ref{fig:invIndex} to build an inverted index for LSH. We treat each hash function as an attribute, and the hash signature as the value for each data point. The keyword in the inverted index for point $p$ under hash function $h_i(\cdot)$ is a pair $(i,h_i(p))$ and the postings list of pair $(i,h_i(p))$ is a set of points whose hash value by $h_i(\cdot)$ is $h_i(p)$ (i.e. $h_i(p')=h_i(p)$ if $p'$ and $p$ in the same postings list.).
Given a query point $q$, we also convert $q$ with the same transformation process, i.e. $Q=[h_1(q), h_2(q),...,h_m(q)]$. After building the inverted index, the query search according to the match-count model can be conveniently supported by GENIE.

As we will prove in Section \ref{sec:search:ann}, the top result returned by GENIE according to the match-count model is just the  ANN search result. Any similarity measures associated with an LSH family defined by Eqn. \ref{eqn:simlsh} can be supported by GENIE for ANN search. For ANN search in high dimensional space, we usually resort to $(r_1,r_2,\rho_1,\rho_2)$-sensitive hashing function family. We give a special discussion about it in Section \ref{sec:search:lpspace}. We will also analyze the error bound between the estimate $\hat{s}$ and the real similarity measure $s=sim(p,q)$ in Section \ref{sec:search:ann}.


\subsubsection{Re-hashing for LSH with large signature space}
A possible problem is that the hash signature of LSH functions may have a huge number of possible values with acceptable error by configuring the LSH parameters. For example, the signature of the Random Binning Hashing function introduced later can be thousands of bits by setting good parameters for search error. Meanwhile, it is not reasonable to discretize the hash signature into a set of buckets according to the definition of LSH in Eqn. \ref{eqn:simlsh}. 

To reduce the number of possible signatures introduced by LSH, we propose a re-hashing mechanism which is illustrated in Figure \ref{fig:rehashing}. After obtaining the LSH signature $h_i(\cdot)$, we further randomly project the signatures into a set of buckets with a random projection function $r_i(\cdot)$. We can convert a point to an object by the transformation: $O_i=[r_1(h_1(p_i)), r_2(h_2(p_i)),...,r_m(h_m(p_i))]$ where $h_j(\cdot)$ is an LSH function and $r_j(\cdot)$ is a random projection function. 
Note that re-hashing is not necessary if the signature space of selected LSH can be configured small enough.

\begin{figure}[htb]
\vspace{-0mm}
\centerline{
\includegraphics[width=0.35\textwidth]{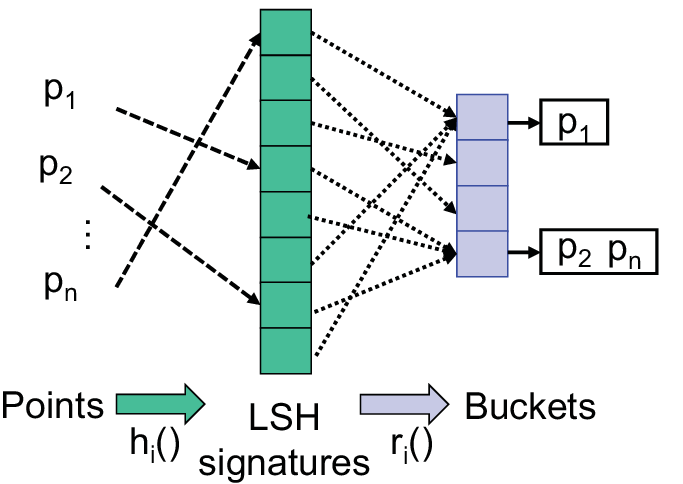}
}
\vspace{-0mm}
 \caption{  Re-hashing mechanism where $h(\cdot)$ is a LSH function and $r(\cdot)$ is a random projection function.}\label{fig:rehashing}
\vspace{-0mm}
\end{figure}

%

\subsubsection{Case study: ANN in Laplacian kernel space} \label{sec:search:case}


We take the ANN search on a shift-invariant kernel space as a case study, which has important applications for machine learning and computer vision. The authors in \cite{rahimi2007random} propose an LSH family, called Random Binning Hashing (RBH), for Laplacian kernel $k(p,q)=exp(-\parallel p-q\parallel_1/\sigma)$.

For any kernel function $k(\cdot)$ satisfying that $p(\sigma)=\sigma k''(\sigma)$ (where $k''(\cdot)$ is the second derivative of $k(\cdot)$) is a probability distribution function on $\sigma\geq 0$, we can construct an LSH family by random binning hashing (RBH) for kernel function $k(\cdot)$. For each RBH function $h(\cdot)$, we impose a randomly shift regular grid on the space with a grid cell size $g$ that is sampled from $p(\sigma)$, and a shift vector $u=[u^1,u^2,...,u^d]$ that is drawn uniformly from $[0,g]$. For a $d$-dimensional point $p=[p^1,p^2,...,p^d]$, the hash function is defined as:
\begin{equation}
h(p)=[\lfloor (p^1-u^1)/g\rfloor,...,\lfloor (p^d-u^d)/g\rfloor]
\end{equation}
The expected collision probability of RBH function is $Pr(h(p)=h(q)]=k(p,q)$. One typical kernel function satisfying the conditions for RBH is Laplacian kernel $k(p,q)=exp(-\parallel p-q\parallel_1/\sigma)$. We refer interested readers to \cite{rahimi2007random}.
For the re-hashing mechanism on GENIE, we select MurmurHashing3 \cite{appleby2015smhasher} as the random projection function.

Though this method is well-known for dimension reduction, as far as we know, it has not been applied to ANN search. One possible reason is that this method requires a huge hash signature space, where the number of bits required is a multiple of the number of dimensions of points.

In experiment, we demonstrate that GENIE can support ANN search in Laplacian kernel space based on RBH. To reduce the hash signature space, we use the re-hashing mechanism to project each signature into a finite set of buckets. 

\subsection{Theoretical analysis }\label{sec:search:ann}


To integrate existing LSH methods into GENIE requires a theoretical bound analysis for LSH under match-count model. For this purpose, we first propose a revised definition of the ANN search, called Tolerance-Approximate Nearest Neighbor search ($\tau$-ANN).
\begin{definition}[$\tau$-ANN]\label{def:tann}
\vspace{-0mm}
\emph{Given a set of $n$ points $P=\{p_1,p_2,..,p_n\}$ in a space $S$ under a similarity measure $sim(p_i,q)$, the Tolerance-Approximate Nearest Neighbor ($\tau$-ANN) search returns a point $p$ such that $|sim(p,q)-sim(p^*,q)|\leq \tau$ with high probability where $p^*$ is the true nearest neighbor.}
\vspace{-0mm}
\end{definition}
This concept is similar to the popular definition of $c$-ANN \cite{indyk1998approximate}.\footnote{$c$-ANN search  is defined as: to find a point $p$ so that $sim(p,q)\leq c\cdot sim(p^*,q)$  with high probability where $p^*$ is the true nearest neighbor.} Some existing works, like \cite{satuluri2012bayesian}, have also used a concept similar to Definition \ref{def:tann} though without explicit definition.

%


\subsubsection{Error bound and $\tau$-ANN}\label{sec:search:errbnd}
We prove that the top return of GENIE for a query $q$ is the $\tau$-ANN of $q$. Given a point $p$ and a query $q$ with a set of LSH functions $\mathbb{H}=\{h_1,h_2,...,h_m\}$, suppose there are $c$ functions in $\mathbb{H}$ satisfying $h_i(p)=h_i(q)$ (where $c$ is just the return of match-count model). We first prove in Theorem \ref{lemma:errBound} that 
the return of match-count model on GENIE can be probabilistically bounded w.r.t the similarity between $p$ and $q$, i.e. $|c/m-sim(p,q)|<\epsilon$ with high probability. 


\begin{theorem}\label{lemma:errBound}
\vspace{-0mm}
\emph{Given a similarity measure $sim(\cdot,\cdot)$, an LSH family $h(\cdot)$, we can get a new hash function $f(x)=r(h(x))$, where $r(\cdot)$ is a random projection function from LSH signature to a domain $R:U\rightarrow [0,D)$.}

\emph{For a set of hash functions $f_i(\cdot)=r_i(h_i(\cdot)),1\leq i\leq m$ with  $m=2\frac{ln (3/\delta)}{\epsilon^2}$, we can convert a point $p$ and a query $q$ to an object and a query of the match-count model, which are $O_p=[f_1(p),f_2(p),...,f_m(p)]$ and $Q_q=[f_1(q),f_2(q),...,f_m(q)]$, then we have $|MC(Q_q, O_p)/m-sim(p,q)|<\epsilon+1/D$ with probability at least $1-\delta$.}
\vspace{-0mm}
\end{theorem}

\begin{proof}
\emph{ The proof of Theorem \ref{lemma:errBound} is inspired by the routine of the proof of Lemma 4.2.2 in \cite{andoni2009nearest}. For convenience, let $c$ be the the return of match-count model $c=MC(O_p,Q_q)$, which essentially is the number of hash function $f_i(\cdot)$ such that $f_i(p)=f_i(q)$. The collisions of $f_i(\cdot)$ can be divided into two classes: one is caused by the collision of the LSH $h_i(\cdot)$ (since if $h_i(p)=h_i(q)$ then we must have $f_i(p)=f_i(q)$), the other one is caused by the collision of the random projection (meaning $h_i(p)\neq h_i(q)$ but $r_i(h_i(p))= r_i(h_i(q))$). Therefore, we further decompose count $c$ as $c=c_h+c_r$ where $c_h$ denotes the number of collisions of caused by $h_i(\cdot)$ and $c_r$ denotes the one caused by $r_i(\cdot)$ when  $h_i(p)\neq h_i(q)$.}

\emph{We first prove that that $|c_h/m-sim(p,q)|\leq \epsilon/2$ with probability at least $1-\delta/2$ given $Pr[h(p)=h(q)]=sim(p,q)$. This is deduced directly by Hoeffding's inequity if $m=2\frac{ln (3/\delta)}{\epsilon^2}$, which is}:
\begin{equation}
\vspace{-0mm}
Pr[|c_h/m-sim(p,q)|\geq \epsilon/2]\leq 2e^{-2m(\frac{\epsilon}{2})^2}=\frac{2\delta}{3}
\vspace{-0mm}
\end{equation}

\emph{For the rest of the $m-c_h$ hashing functions, we need to prove that the collision $c_r\leq (\omega+\epsilon/2)m$ with probability at least $1-\delta/3$, where $\omega$ is the collision probability of $r_i$. To simplify the following expression, we also denote $\epsilon/2$ by $\beta$, i.e. $\beta = \epsilon/2$ . Note that the expectation of $c_r$ is $E(c_r)=\omega(m-c_h)$. According to the Hoeffding's inequality, we have $Pr[c_r>(\omega+\beta)m]=Pr[c_r>(\omega+\frac{\omega c_h+\beta m}{m-c_h})(m-c_h)]\leq e^{-2(\frac{\omega c_h+\beta m}{m-c_h})^2(m-c_h)}$. We have $(\frac{\omega c_h+\beta m}{m-c_h})^2(m-c_h)\ge (\frac{\beta m}{m-c_h})^2(m-c_h)\ge \frac{\beta^2 m^2}{m-c_h}\ge \beta^2m$. Therefore, we have $e^{-2(\frac{\omega c_h+\beta m}{m-c_h})^2(m-c_h)}\leq e^{-2\beta^2m}$. Finally, we have $Pr[c_r/m>(\omega+\beta)]\leq e^{-2\beta^2m}$. Since $\beta=\epsilon/2$ and $m=2\frac{ln (3/\delta)}{\epsilon^2}$, we have $Pr[c_r/m>(\omega+\epsilon/2)]\leq \frac{\delta}{3}$. To combine above together, we have}:

\begin{align*}
\vspace{-0mm}
&Pr[|c_h/m+c_r/m-sim(p,q)|>\omega+\epsilon] \\
&\leq Pr[(|c_h/m-sim(p,q)|>\epsilon/2) \cup (c_r/m>\omega+\epsilon/2)]\\
&\leq Pr[(|c_h/m-sim(p,q)|>\epsilon/2)]+ Pr[(c_r/m>\omega+\epsilon/2)]\\
&\leq \delta
\vspace{-0mm}
\end{align*}

\emph{We also have $\omega=D*1/D^2=1/D$, therefore, $|MC(O_p,Q_q)-sim(x,y)|=|(c_h+c_r)/m-sim(p,q)|\leq \epsilon+1/D$ with probability at least $1-\delta$.}
\end{proof}

Now we introduce an important theorem which claims that, given a query point $q$ and proper configuration of $m$ stated in Theorem \ref{lemma:errBound}, the top result returned by GENIE is just the $\tau$-ANN of $q$.

\begin{theorem}\label{lemma:tann}
\vspace{-0mm}
\emph{Given a query $q$ and a set of points $P=\{p_1,p_2,..,p_n\}$, we can convert them to the objects of our match-count model by transformation $O_{p_i}$ $=$ $[r_1(h_1(p_i))$, \\ $r_2(h_2(p_i))$ $,...,$ $r_m(h_m(p_i))]$
which satisfies $|MC(Q_q, O_{p_i})/m-sim(p_i,q)|$ $\leq \epsilon$ with the probability at least $1-\delta$.  Suppose the true NN of $q$ is $p^*$, and the top result based on the match-count model is $p$, then we have $|sim(p^*,q)-sim(p,q)|\leq 2\epsilon$ with probability at least $1-2\delta$.}
\vspace{-0mm}
\end{theorem}
\begin{proof}
\emph{For convenience, we denote that the output count values of match-count model as $c=MC(O_p,Q_q)$ and $c^*=MC(O_{p^*},Q_q)$, and denote the real similarity measures as $s=sim(p,q)$ and $s^*=sim(p^*,q)$. We can get that}
\begin{align*}
\vspace{-0mm}
&Pr[|c/m-s|\leq \epsilon \cap |c^*/m-s^*|\leq \epsilon]\\
&=Pr[|c/m-s|\leq \epsilon]\cdot Pr[|c^*/m-s^*|\leq \epsilon]\\
&\geq (1-\delta)(1-\delta)\geq 1-2\delta
\vspace{-0mm}
\end{align*}

\emph{We also have that $c\geq c^*$ ($c$ is top result) and $s^*\geq s$ ($s^*$ is true NN). From $|c/m-s|\leq \epsilon$ and $|c^*/m-s^*|\leq \epsilon$, we can get that $s^*\leq c^*/m+\epsilon$ and $s\geq c/m-\epsilon$, which implies that $s^*-s\leq c^*/m+\epsilon-(c/m-\epsilon)\leq 2\epsilon$. To sum up, we can obtain that $Pr[sim(p^*,q)-sim(p,q)\leq 2\epsilon]\geq 1-2\delta$}
\end{proof}

\subsubsection{Number of hash functions in practice}
Theorem \ref{lemma:errBound} provides a rule to set the number of LSH functions 
as $O(\frac{1}{\epsilon^2})$ which may be very large. It is NOT a problem for GENIE to support such a number of hash functions since the GPU is a parallel architecture suitable for the massive quantity of simple tasks. The question however is that: Do we really need such a large number of hash functions in practical applications?

Before exploiting this, we first explain that the collision probability of a hash function $f_i(\cdot)$ can be approximated with the collision probability of an LSH function $h_i(\cdot)$ if $D$ is large enough. 
The collision probability of $f_i(\cdot)$ can be factorized as collisions caused by $h_i(\cdot)$ and collisions caused by $r_i(\cdot)$:
\begin{align}
&Pr[f_i(p)=f_i(q)]=Pr[r_i(h_i(p))=r_i(h_i(q))]\\
&\leq Pr[h_i(p)=h_i(q)] + Pr[r_i(h_i(p))=r_i(h_i(q))]\\
&=s+1/D
\end{align}
where $s=sim(p,q)$. Thus, we have $s\leq Pr[f_i(p)=f_i(q)]\leq s+1/D$. 
Suppose $r(\cdot)$ can re-hash $h_i(\cdot)$ into a very large domain $[0,D)$, we can claim that $Pr[f_i(p)=f_i(q)]\thickapprox s$. For simplicity, let us denote $c=MC(Q_q, O_p)$. An estimation of
$s$ by maximum likelihood estimation (MLE) can be \cite{satuluri2012bayesian}:
\begin{equation}\label{eqn:appro_mc}
s=MC(Q_q, O_p)/m=c/m
\end{equation}


Eqn. \ref{eqn:appro_mc} can be further justified by the following equation:
\begin{align}
Pr[|\frac{c}{m}-s|\leq \epsilon] &= Pr[(s-\epsilon)*m\leq c \leq (s+\epsilon)*m]\\
&=\sum_{c=\lfloor (s-\epsilon)m \rfloor}^{\lceil (s+\epsilon)m \rceil} {m\choose c}s^{c}(1-s)^{m-c}\label{eqn:err_hat_s}
\end{align}
Eqn. \ref{eqn:err_hat_s} shows that the probability of error bound depends on the similarity measure $s=sim(p,q)$  \cite{satuluri2012bayesian}. Therefore, there is no closed-form expression for such error bound.


Nevertheless, Eqn. \ref{eqn:err_hat_s} provides a practical solution to estimate a tighter error bound of the match-count model different from Theorem \ref{lemma:errBound}. If we fixed $\epsilon$ and $\delta$, for any given similarity measure $s$, we can infer the number of required hash functions $m$ subject to the constraint $Pr[|c/m-s|\leq \epsilon]\geq 1-\delta$ according to Eqn. \ref{eqn:err_hat_s}. Figure \ref{fig:s_vs_m} visualizes the number of minimum required LSH functions for different similarity measure with respect to a fixed parameter $\epsilon=\delta=0.06$ by this method.  A similar figure has also been illustrated in \cite{satuluri2012bayesian}. As we can see from Figure \ref{fig:s_vs_m}, the largest required number of hash functions, being m=237, appears at $s=0.5$,  which is much smaller than the one estimated by Theorem \ref{lemma:errBound} (which is $m=2\frac{ln (3/\delta)}{\epsilon^2}=2174$). We should note that the result shown in Figure \ref{fig:s_vs_m} is data independent. Thus, instead of using Theorem \ref{lemma:errBound}, we can effectively estimate the actually required number of LSH functions using the simulation result based on Eqn. \ref{eqn:err_hat_s} (like Figure \ref{fig:s_vs_m}).

\begin{figure}[htb]
\vspace{-0mm}
\centerline{
\includegraphics[width=0.4\textwidth]{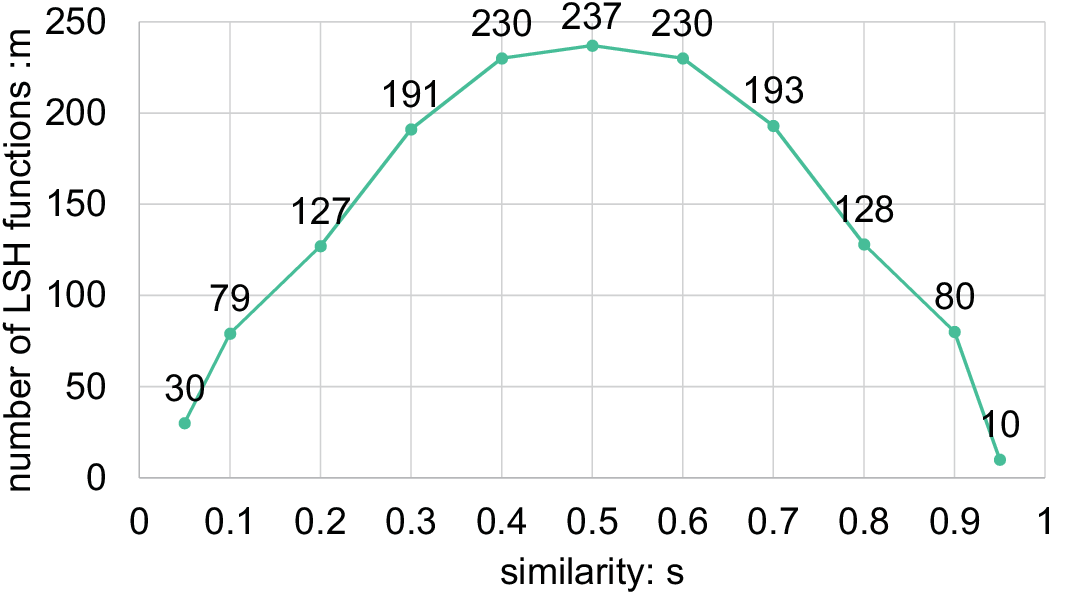}
}
\vspace{-0mm}
 \caption{  Similarity (s) v.s. the number of minimum required LSH functions (m) with constraint $Pr[|c/m-s|\leq \epsilon]\geq 1-\delta$ where $\epsilon=\delta=0.06$. 
 }
\label{fig:s_vs_m}
\vspace{-0mm}
\end{figure}

\subsubsection{ANN search in high dimensional space}\label{sec:search:lpspace}

For ANN search in high dimensional space, we usually resort to $(r_1,r_2,\rho_1,\rho_2)$-sensitive hashing function family which deserves special discussion. The LSH function family of $(r_1,r_2,\rho_1,\rho_2)$-sensitive hashing is usually defined as follows \cite{datar2004locality}:
\begin{definition}
\emph{In a $d$-dimensional $l_\mathfrak{p}$ norm space $R^d$, a function family $\mathbb{H}=\{h:R^d\rightarrow [0,D)\}$ is called $(r_1,r_2,\rho_1,\rho_2)$-sensitive if for any $p,q\in R^d$:}
\begin{itemize}\label{def:e2lsh}
\item if $\parallel p-q\parallel_\mathfrak{p}\leq r_1, then\ Pr[h(p)=h(q)]\geq \rho_1$
\item if $\parallel p-q\parallel_\mathfrak{p}\geq r_2, then\ Pr[h(p)=h(q)]\leq \rho_2$
\end{itemize}
\emph{where $r_1<r_2$ and $\rho_1>\rho_2$ and  $\parallel p-q\parallel_\mathfrak{p}$ is distance function in $l_\mathfrak{p}$ norm space.}
\end{definition}
Based on the $\mathfrak{p}$-stable distribution \cite{indyk2000stable}, an LSH function family in $l_\mathfrak{p}$ norm space can be defined as \cite{datar2004locality}:
\begin{equation}\label{eqn:plsh}
\vspace{-0mm}
h(q)=\lfloor \frac{\mathbf{a}^T\cdot q+b}{w}\rfloor
\vspace{-0mm}
\end{equation}
where $\mathbf{a}$ is a d-dimensional random vector whose entry is drawn independently from a $\mathfrak{p}$-stable distribution for $l_\mathfrak{p}$ distance function (e.g. Gaussian distribution for $l_\mathfrak{2}$ distance), and $b$ is a random real number drawn uniformly from $[0,w)$.

In order to integrate such an LSH function into our proposed match-count model for ANN search, we have to find the relation between the collision probability and the $l_\mathfrak{p}$ distance. For this purpose, we justify the LSH function of Eqn. \ref{eqn:plsh} by the following equation \cite{datar2004locality} (let $\Delta=\parallel p-q\parallel_\mathfrak{p})$:
\begin{equation}\label{eqn:psi}
\psi_\mathfrak{p}(\Delta)=Pr[h(p)=h(q)]=\int_0^w\frac{1}{\Delta}\phi_\mathfrak{p}(\frac{t}{\Delta})(1-\frac{t}{w})dt
\end{equation}
where $\phi_\mathfrak{p}(\Delta)$ denotes the probability distribution density function of the absolute value of the $\mathfrak{p}$-stable distribution. 

From Eqn. \ref{eqn:psi}, we can infer that $\psi_\mathfrak{p}(\Delta)$ is a strictly monotonically decreasing function \cite{datar2004locality}: If $p_1$ is more nearby to $q$ than $p_2$ in $l_p$ space ($\Delta_1=\parallel p_1-q\parallel_\mathfrak{p}< \Delta_2=\parallel p_2-q\parallel_\mathfrak{p}$), then $\psi_\mathfrak{p}(\Delta_1)$ is higher than $\psi_\mathfrak{p}(\Delta_2)$.
Therefore, we can say that $\psi_\mathfrak{p}(\Delta)$ defines a similarity measure between two points in  $l_p$ norm space, i.e.
\begin{equation}\label{eqn:simlp}
sim_{l_\mathfrak{p}}(p,q)=\psi_\mathfrak{p}(\Delta)=\psi_\mathfrak{p}(\parallel p-q||_\mathfrak{p}).
\end{equation}

Recalling Theorem \ref{lemma:errBound} and Theorem \ref{lemma:tann}, the only requirement for LSH functions of GENIE is to satisfy Eqn. \ref{eqn:simlsh}, which can be justified by Eqn. \ref{eqn:psi} for $(r_1,r_2,\rho_1,\rho_2)$-sensitive hashing function family. In other words, we can use GENIE to do $\tau$-ANN search under the similarity measure of Eqn. \ref{eqn:simlp}. 
Though the ANN search result is not measured by the $l_\mathfrak{p}$ norm distance, the returned result follows the same criterion to select the nearest neighbor since the similarity measure defined in Eqn.\ref{eqn:simlp} is closely related to the $l_p$ norm distance.



A similar counting method has also been used for $c$-ANN search in \cite{gan2012locality} where a Collision Counting LSH (C2LSH) scheme is proposed for $c$-ANN search. 
Though our method has different theoretical perspective from C2LSH, the basic idea behind them is similar: the more collision functions between points, the more likely that they would be near each other. From this view, the C2LSH can corroborate the effectiveness of the ANN search of GENIE.

\section{Searching on data with SA}\label{sec:search_org}

GENIE also provides a choice of adopting the ``Shotgun and Assembly'' (SA) scheme for similarity search. Given a dataset, we split each object into small units. Then we build inverted index where each unique unit is a keyword, and the corresponding postings list is a list of objects containing this unique unit. When a query comes, it is also broken down as a set of such small units. After that, GENIE can effectively calculate the number of common units between the query object and data objects. 

The return of match-count model can either be a similarity measure (e.g. document search where the count is just the inner product between the space vector of documents), or be considered as a lower bound of a distance (e.g. edit distance) to filter candidates \cite{wang2013efficient,yang2005similarity}.
We will demonstrate how to perform similarity search on sequence data, short document data and relational data using GENIE.

\subsection{Searching on sequence data}\label{sec:search:sequence}
In this section, we show how to use GENIE to support similarity search by the SA scheme with an example of sequence similarity search under edit distance. The general process is to first chop the data sequences and query sequences as $n$-grams using sliding windows, and then to build inverted index on the GPU with the $n$-grams as keywords. 

\subsubsection{Shotgun: decomposition and index}
We first decompose the sequence $S$ into a set of $n$-grams using a length-$n$ sliding window. Given a sequence $S$ and an integer $n$, the $n$-gram is a length-$n$ subsequence $s$ of $S$. Since the same $n$-gram may appear multiple times in a sequence, we introduce the \emph{ordered $n$-gram}, which is a pair ($n$-gram,$i$) where $i$ denotes the $i$-th same $n$-gram in the sequence. Therefore, we decompose the sequence $S$ into a set of ordered $n$-gram $G(S)$.  In GENIE, we build an inverted index by treating the ordered $n$-gram as a keyword and putting its sequence id in the postings list.

\begin{example}
\vspace{-0mm}
\emph{For a sequence $S=\{aabaab\}$, the set of \emph{ordered 3-grams} of $S$ is $G(S)=\{(aab,0),(aba,0),(baa,0),(aab,1)\}$ where $(aab,0)$ denotes the first subsequence $aab$ in $S$, and $(aab,1)$ denotes the second subsequence $aab$ in $S$.}
\vspace{-0mm}
\end{example}

\subsubsection{Assembly: combination and verification}
After building the index, we decompose a query sequence $Q$ into a set of \emph{ordered $n$-grams} using sliding windows, i.e. $G(Q)$. GENIE can then retrieve candidates with top-$k$ large count in the index. Before explaining this method, we first propose the following lemma:
\begin{lemma}
\vspace{-0mm}
\emph{Suppose the same $n$-gram $s_n^i$ appears $c_s^i$ times in sequence $S$ and $c_q^i$ times in sequence $Q$, then the result returned by the match-count model is $MC(G(S),G(Q))=\sum_{s_n^i}min\{c_s^i,c_q^i\}$.}
\vspace{-0mm}
\end{lemma}

 With respect to the edit distance, the result of the match-count model satisfies the following theorem \cite{sutinen1996filtration}.
\begin{theorem}\label{lemma:edit_bound}
\vspace{-0mm}
\emph{If the edit distance between $S$ and $Q$ is $\tau$, then the return of the match-count model has $MC(G(S),G(Q))\ge \max\{|Q|,|S|\}-n+1-\tau*n$. \cite{sutinen1996filtration}}
\vspace{-0mm}
\end{theorem}
\begin{proof}
\emph{The observation between the count of shared $n$-gram and edit distance has been studied in \cite{sutinen1996filtration,gravano2001approximate}. Note that if an n-gram occurs $c_s$ times in $S$ and $c_q$ times in $Q$, we count the common $n$-gram number as $min(c_s,c_q)$ \cite{karkkainen2002computing}.}
\end{proof}


According to Theorem \ref{lemma:edit_bound}, we can use the result of the match-count model as an indicator for selecting candidates for the query sequence.  Our strategy is to retrieve $\mathbf{K}$ candidates from GENIE according to match count with setting a large $\mathbf{K} (\mathbf{K}>>k)$. Then  we can employ a verification process to calculate the edit distance between the query $Q$ and the $\mathbf{K}$ candidates to obtain the $k$-th most similar sequence $S^{\mathbf{k'}}$. The algorithm of the verification process is shown in Algorithm
\ref{alg:seq_verification}.

{\vspace{-0mm}
\begin{algorithm}[!htb]
\small
\DontPrintSemicolon
\SetKwInOut{Input}{Input}\SetKwInOut{Output}{Output}

\Input{$\mathbf{K}$ candidates of $Q$ returned by GENIE, with the descent order according to the count.}

\Output{The most similar sequence among $\mathbf{K}$ candidates}

\ShowLn $\tau_*=ed(Q,S^1)$ \tcp{$\tau^*$ is the smallest edit distance.}

\ShowLn $S^*=S^1$ \tcp{$S^*$ is the best candidate.}

\ShowLn $\theta =|Q|-n+1-n(\tau_*-1)$ \tcp{filtering bound}

\For{$2\leq j\leq \mathbf{K}$}{

\If{$\theta>MC(Q,S^j)$}{

\ShowLn break \tcp{We have $\tau_*<\tau_j$ by Theorem \ref{lemma:edit_bound}.}

\If{$|Q|-|S^j|>\tau_*$ \tcp{length filter}}
{

\ShowLn continue

}
\ShowLn $\tau_j=ed(Q,S^j)$

\If{$\tau_j<\tau_*$}{

\ShowLn $\tau_*=\tau_j$

\ShowLn $S^*=S^j$

\ShowLn $\theta=|Q|-n+1-n(\tau_*-1)$

}
}
}
 \caption{Verification algorithm for sequence search}\label{alg:seq_verification}

\end{algorithm}
 \vspace{-0mm}
}

With this method, we can know whether the real top-$k$ sequences are correctly returned by GENIE, though we cannot guarantee the returned top-$k$ candidates are the real top-$k$ data sequence for all queries. 
In other words, after the verification, we can know whether $S^{\mathbf{k'}}$ is the real $k$-th most similar sequence of Q accroding to the following theorem. 
\begin{theorem}\label{lemma:top-$k$}
\vspace{-0mm}
\emph{For the $\mathbf{K}$-th candidates $S^\mathbf{K}$ returned by GENIE according to count, suppose the match count between $S^\mathbf{K}$  and query $Q$ is $c_{\mathbf{K}}=MC(G(S^\mathbf{K}),G(Q))$. Among the $\mathbf{K}$ candidates, after employing the verification algorithm, we can obtain the edit distance between $k$-th most similar sequence (among the $\mathbf{K}$ candidates) and $Q$ is $\tau_{\mathbf{k'}}=ed(Q,S^{\mathbf{k'}})$. If $c_{\mathbf{K}}<|Q|-n+1-\tau_{\mathbf{k'}}*n$, then the real top-$k$ results are correctly returned by GENIE.}
\vspace{-0mm}
\end{theorem}

\begin{proof}
\emph{We prove it by contradiction. Suppose the real k-th similar sequence $S^{k*}$ cannot be found by GENIE. The edit distance between the query $Q$ and the real k-th similar sequence $S^{k*}$ is $\tau_{k^{*}}$. Then the match count between $S^{k*}$ and $Q$ satisfies: $c_{k^*}\geq max\{|Q|,|S^{k*}|\} -n+1-\tau_{k^*}*n$. We also have $\tau_{k^*}\leq \tau_{\mathbf{k'}}$, then we can get   $c_{k^*}\geq |Q|-n+1-\tau_{k^*}*n \geq |Q|-n+1-\tau_{\mathbf{k'}}*n >c_{\mathbf{K}}$. However, if $c_{k^*}>c_{\mathbf{K}}$, sequence $S^{k*}$ must have already been retrieved by GENIE. Therefore, the real k-th similar sequence $S^{k*}$ can be found by GENIE. In the same way, we can prove that all the top-$k$ results can be correctly returned by GENIE under the premise condition.}
\end{proof}

A possible solution for sequence similarity search is to repeat the search process by GENIE with larger $\mathbf{K}$, until the condition in Lemma \ref{lemma:top-$k$} is satisfied. In the worst case it may need to scan the whole data set before retrieving the real top-$k$ sequences under edit distance. 
However, as shown in our experiment, it can work well in practice for near edit distance similarity search in some applications. 


\subsection{Searching on short document data}\label{sec:search:string}

In this application, both query documents and object documents are broken down into ``words''. We build an inverted index with GENIE where the keyword is a ``word'' from the document, and the postings list is a list of document ids. 

We can explain the result returned by GENIE on short document data by the document vector space model. Documents can be represented by a binary vector space model where each word represents a separate dimension in the vector. If a word occurs in the document, its value in the vector is one, otherwise it is zero.  The output of the match-count model, which is the number of word co-occurring in both the query and the object,  is just the \emph{inner product} between the binary sparse vector of the query document and the one of the object document. 


\subsection{Searching on relational data}\label{sec:search:relation}
GENIE can also be used to support queries on relational data. In Fig. \ref{fig:invIndex}, we have shown how to build an inverted index for relational tuples. For attributes with continuous value, we assume that they can be discretized to an acceptable granularity level. A range selection query on a relational table is a set of specific ranges on attributes of the relational table. 

The top-$k$ result returned by GENIE on relational tables can be considered a special case of the traditional top-$k$ selection query. The top-$k$ selection query selects the $k$ tuples in a relational table with the largest predefined ranking score function $F(\cdot)$ (SQL \emph{ORDER BY}  $F(\cdot)$). 
In GENIE, we use a special ranking score function defined by the match-count model, which is especially useful for tables having both categorical and numerical attributes.

\section{Experiments}\label{sec:exp}

\subsection{Settings}

\subsubsection{Datasets}

We use five real-life datasets to evaluate our system. Each dataset corresponds to one similarity measure respectively introduced in Section \ref{sec:search} and Section \ref{sec:search_org}.

$\left[\textbf{OCR}\right]$\footnote{ \mbox{  \url{http://largescale.ml.tu-berlin.de/instructions/}} }
This is a dataset for optical character recognition. It contains 3.5M data points and each point has 1156 dimensions. The size of this dataset is 3.94 GB. We randomly select 10K points from the dataset as query/test set (and remove them from the dataset). We use RBH to generate the LSH signature, which is further re-hashed into an integer domain of [0,8192) (see Section \ref{sec:search:case}). 

$\left[\textbf{SIFT}\right]$\footnote{ \mbox{  \url{http://lear.inrialpes.fr/~jegou/data.php}} }
This dataset \cite{jegou:2008hamming} contains 4.5M SIFT features which are 128-dimensional points. Its total size is 1.49 GB. We randomly select 10K features as our query set and remove them from the dataset. We select the hash functions from the E2LSH family \cite{datar2004locality} and each function transforms a feature into 67 buckets. The setting of the bucket width follows the routine in \cite{datar2004locality}. We use this dataset to evaluate the ANN search in high dimensional space. 

$\left[\textbf{SIFT\_LARGE}\right]$ \footnote{ \mbox{  \url{http://image-net.org/challenges/LSVRC/2010/index#data}} } To evaluate the scalability of our system, we also extract 36 millions SIFT features by ourselves from the ILSVRC-2010 image dataset. The size of this dataset is 14.75 GB. We use the same method as described above for SIFT to process the data. 

$\left[\textbf{DBLP}\right]$\footnote{ \mbox{  \url{http://dblp.uni-trier.de/xml/}} }
 This dataset is obtained by extracting article titles from the DBLP website. The total number of sequences is 5.0M and the size of this dataset is 0.94 GB. We randomly choose 10K sequences as a test data, and then modify 20\% of the characters of the sequences. 
 This dataset is  to serve the experiment of sequence similarity search in Section \ref{sec:search:sequence}. Specially, 
 we set $\mathbf{K}=32$ and $k=1$ by default.

$\left[\textbf{Tweets}\right]$\footnote{ \mbox{  \url{https://dev.twitter.com/rest/public}} } This dataset has 6.8M tweets. We remove stop words from the tweets. The dataset is crawled by our collaborators from Twitter for three months by keeping the tweets containing a set of keywords.\footnote{   The keywords include ``Singapore'', ``City'', ``food joint'' and ``restaurant'', etc. It is crawled for a research project. } 
The data size is 0.46 GB.
We reserve 10K tweets as a query set. It is used to study the short document similarity search (see Section \ref{sec:search:string}).

$\left[\textbf{Adult}\right]$\footnote{ \mbox{  \url{http://archive.ics.uci.edu/ml/datasets/Adult}} }
This dataset has census information \cite{lichman2013} which contains 49K rows with 14 attributes (mixed of numerical and categorical ones). For numerical data, we discretize all value into 1024 intervals of equal width. We further duplicate every row 20 times. Thus, there are 0.98M instances (with size being 5.8 GB).  We select 10K tuples as queries. For numerical attributes, the query item range is defined as $[discretized\_value - 50, discretized\_value + 50]$. We use it to study the selection from relational data (see Section \ref{sec:search:relation}). 

\subsubsection{Competitors}\label{sec:exp:baseline}
We use the following competitors as baselines to evaluate the performance of GENIE.

$\left[\textbf{GPU-LSH}\right]$
We use GPU-LSH \cite{pan2011fast,pan2012bi} as a competitor of GENIE for ANN search in high dimensional space  and its source code is publicly available\footnote{\mbox{  \url{http://gamma.cs.unc.edu/KNN/}}}. Furthermore, since there is no GPU-based LSH method for ANN search in Laplacian kernel space, we still use GPU-LSH method as a competitor for ANN search of GENIE. 
We configure the parameters of GPU-LSH to make sure its ANN search results have similar quality as GENIE, which is discussed in Section \ref{sec:exp:ann}.
We only use 1M data points for GPU-LSH on OCR dataset since GPU-LSH cannot afford more data points.

$\left[\textbf{GPU-SPQ}\right]$
We implemented a priority queue-like method on GPU as a competitor. We first scan the whole dataset to compute match-count values between queries and all points, and store these computed results in an array. Then we use a GPU-based fast k-selection \cite{alabi2012fast} method to extract the top-$k$ candidates from the array for each query. We name this top-$k$ calculation method as  SPQ (which denotes GPU fast k-\underline{s}election from an array as a \underline{p}riority \underline{q}ueue).
We give an introduction to SPQ in Appendix \ref{sec:kselection}.
Note that for ANN search, we scan on the LSH signatures (not original data). 

$\left[\textbf{CPU-Idx}\right]$
We implemented an inverted index on the CPU memory. While accessing the inverted index in memory, we use an array to record the value of match-count model for each object. Then we use a partial quick selection function (with $\Theta(n + k log n)$ worst-case performance) in C++ STL to get the $k$ largest-count candidate objects. 

$\left[\textbf{CPU-LSH}\right]$
We use CPU-LSH \cite{gan2012locality} for ANN search in high dimensional space as a competitor after obtaining its source code from authors' website\footnote{\mbox{  \url{http://ss.sysu.edu.cn/~fjl/c2lsh/C2LSH_Source_Code.tar.gz}}}. We use the suggestion method in the paper to determine the parameters. 

$\left[\textbf{AppGram}\right]$ This is one of the state-of-the-art methods for sequence similarity search under edit distance on the CPU \cite{wang2013efficient}. We use AppGram as a baseline for comparing the running time of GENIE for sequence similarity search. Note that AppGram and GENIE are not completely comparable, since AppGram tries its best to find the true kNNs, while GENIE only does one round search process in the experiment. Thus some true kNNs may be missed (though we know which queries do not have true top-$k$, and another search process can be issued to explore the true kNNs). We give more discussion about this in Section \ref{sec:exp:seq_search}.

$\left[\textbf{GEN-SPQ}\right]$
This is a variant of GENIE but using SPQ instead of c-PQ. We still build inverted index on the GPU for each dataset. 
However, instead of using c-PQ (see Section \ref{sec:index:ch}), we use SPQ (which is the same with the one for GPU-SPQ) to extract candidates from the Count Table. 

\subsubsection{Environment}
We conducted the experiments on a CPU-GPU platform. The GPU used is NVIDIA GeForce GTX TITAN X with 12 GB memory. GPU codes were implemented with CUDA 7.  Other programs were in C++ on CentOS 6.5 server (with 64 GB RAM). The CPU used is Intel Core i7-3820.

Unless otherwise specified, we set $k=100$ 
and set the submitted query number per batch to the GPU as 1024. All the reported results are the average of running results of ten times. By default, we do not enable the load balance function since it is not necessary when the query number is large for one batch process. We also present the experiment about load balance in Section \ref{sec:load_balance}. 
For ANN search, we use the method introduced in Section \ref{sec:search:errbnd} to determine the number of LSH hash functions with setting $\epsilon=\delta=0.06$, therefore the number of hash functions is $m=237$. 

\subsection{Efficiency of GENIE}\label{exp:genie_efficiency}
\subsubsection{Search time for multiple queries}\label{exp:search_time}
We compare the running time among GENIE and its competitors. We do not include index building time for all competitors since index building can be done offline. The index building time of GENIE is discussed in Section \ref{exp:timeprofile}.


\begin{figure*}[tb]
\centering
\begin{tabular}{ccccc}
  \hspace{-5mm}\includegraphics[width=0.2\textwidth]{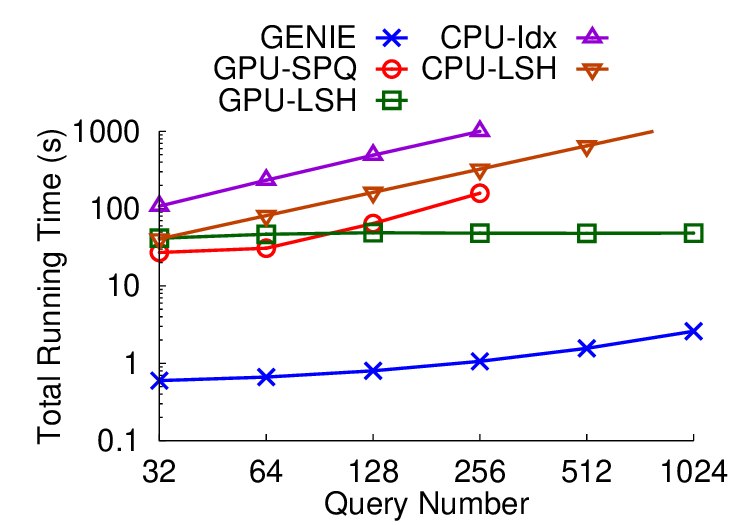} &
     \hspace{-5mm}\includegraphics[width=0.2\textwidth]{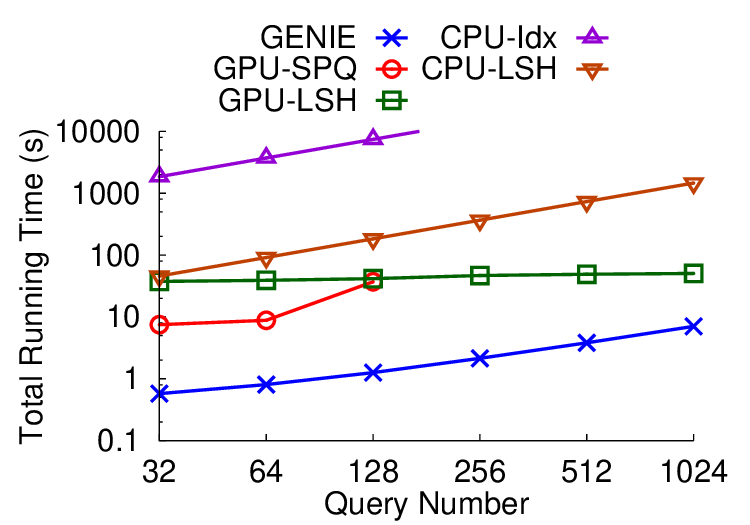}  &
      \hspace{-5mm}\includegraphics[width=0.2\textwidth]{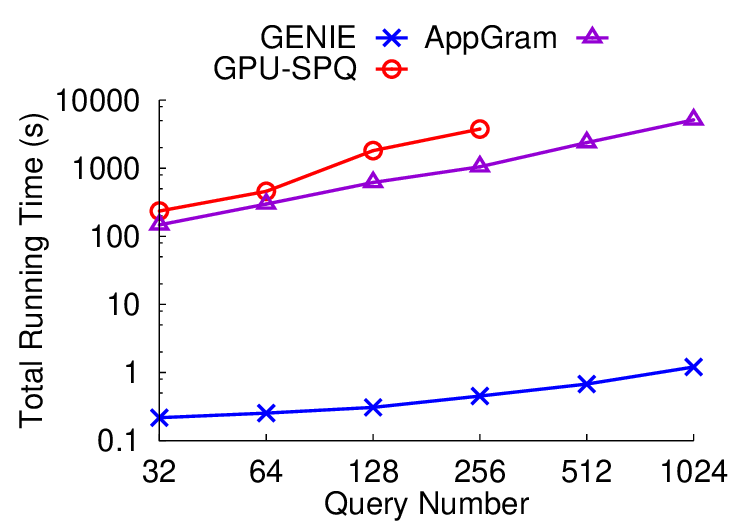} &
      \hspace{-5mm}\includegraphics[width=0.2\textwidth]{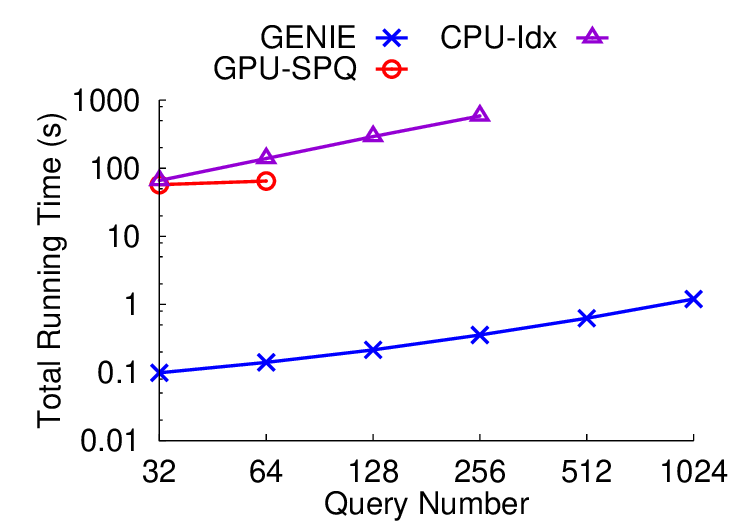}  &
     \hspace{-5mm}\includegraphics[width=0.2\textwidth]{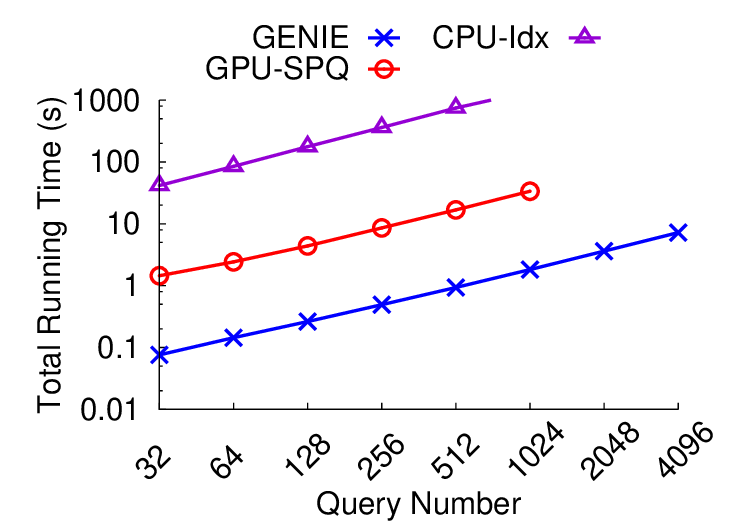}  \\
 \vspace{-0mm}   (a) OCR & \hspace{-0.7cm}    (b) SIFT &   (c) DBLP &   (d) Tweets & \hspace{-0.7cm}   (e) Adult   \\
\end{tabular}
\vspace{-0mm}
\caption{  Total running time for multiple queries.}
\label{fig:exp:total_time_vs_num_query}
\vspace{-0mm}
\end{figure*}

We show the total running time with respect to different numbers of queries in Figure \ref{fig:exp:total_time_vs_num_query} (y-axis is log-scaled). 
Our method outperforms GPU-SPQ by more than one order of magnitude, and it can achieve more than two orders of magnitude over GPU-SPQ and AppGram for sequence search. 
Furthermore, GPU-SPQ can only run less than 256 queries in parallel (except for Adult dataset) for one batch process, but GENIE can support more than 1000 queries in parallel.

As we can see from Figure \ref{fig:exp:total_time_vs_num_query}, GENIE can also outperform GPU-LSH about one order of magnitude. The running time of GPU-LSH is relatively stable with varying numbers of queries. This is because GPU-LSH uses one thread to process one query, thus, GPU-LSH achieves its best performance when there are 1024 queries (which is the maximum number of threads per block on the GPU). 
Note that we only use 1M data points for GPU-LSH on OCR dataset.   

\begin{figure*}[tb]
\centering
\begin{tabular}{ccccc}
  \hspace{-5mm}\includegraphics[width=0.2\textwidth]{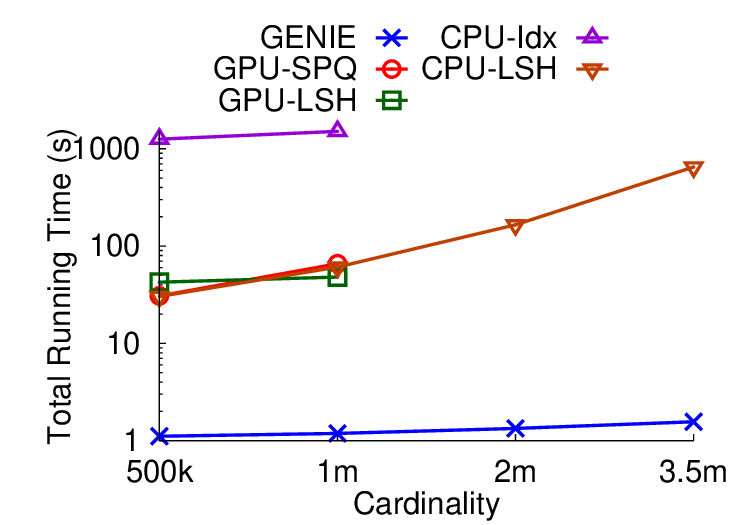} &
     \hspace{-5mm}\includegraphics[width=0.2\textwidth]{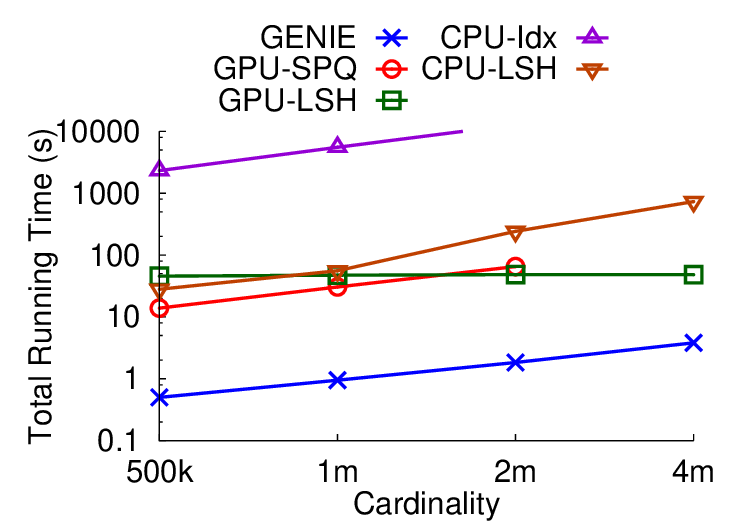}  &
      \hspace{-5mm}\includegraphics[width=0.2\textwidth]{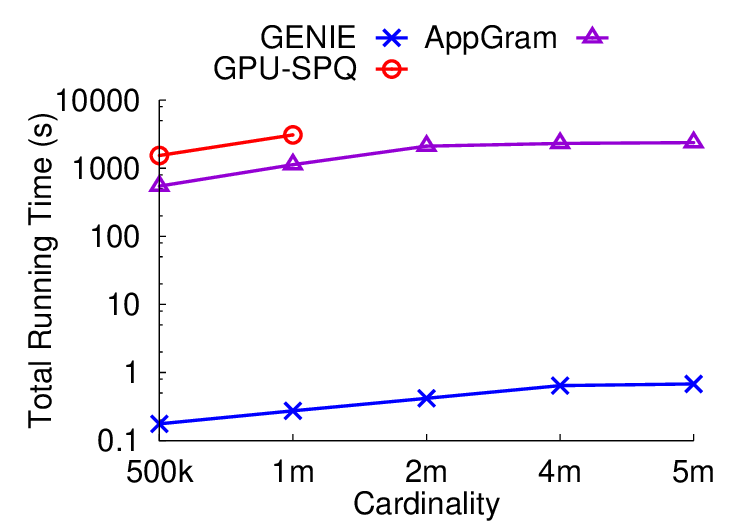}  &
      \hspace{-5mm}\includegraphics[width=0.2\textwidth]{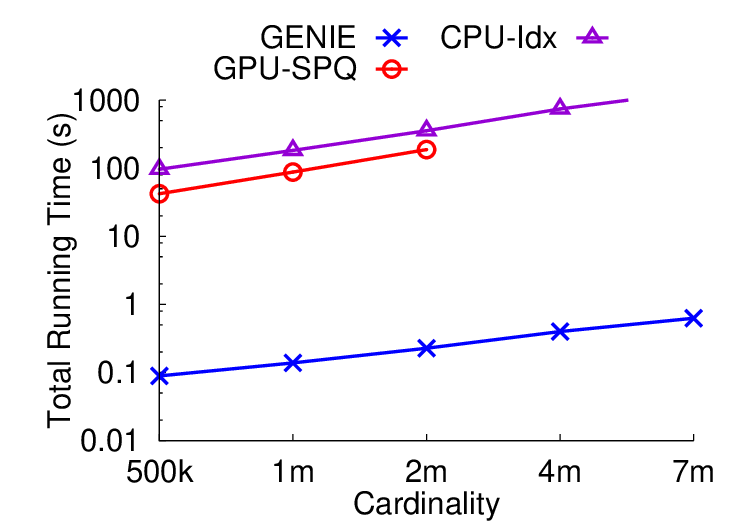}  &
     \hspace{-5mm}\includegraphics[width=0.2\textwidth]{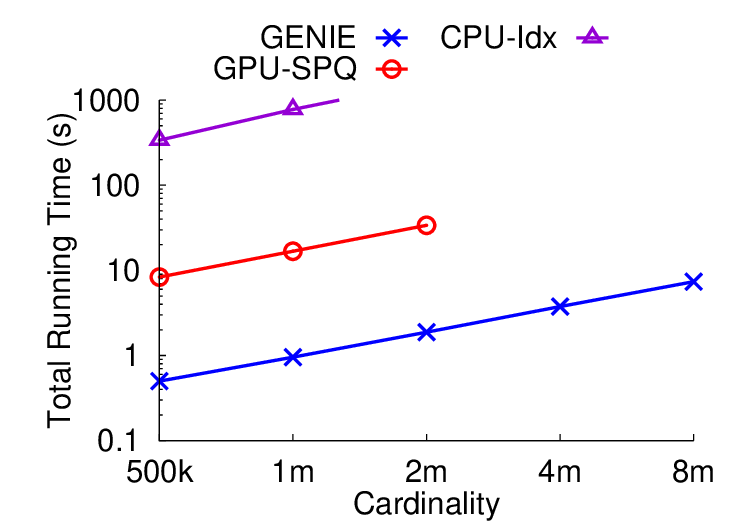} \\
 \vspace{-0mm} \vspace{-1mm}  (a) OCR & \hspace{-0.7cm}    (b) SIFT &   (c) DBLP & (d) Tweets & \hspace{-0.7cm}   (e) Adult   \\
\end{tabular}
\vspace{-0mm}
\caption{  Varying data size for multiple queries (The query number is 512). 
}
\label{fig:exp:total_time_vs_num_data}
\vspace{-0mm}
\end{figure*}

Fig. \ref{fig:exp:total_time_vs_num_data} conveys the running time of GENIE and its competitors with varying numbers of data points for each dataset. 
Since most of the competitors cannot run 1024 queries for one batch, we fix the query number as 512 in this experiment. The running time of GENIE is gradually increased with the growth of data size. 
Nevertheless, the running time of GPU-LSH is relatively stable on all datasets with respect to the data size. The possible reason is that GPU-LSH uses many LSH hash tables and LSH hash functions to break the data points into short blocks, therefore, the time for accessing the LSH index on the GPU becomes the main cost of query processing. 

\begin{figure}[htb]
\vspace{-0mm}
\centering
\begin{minipage}{.4\textwidth}
    \centerline{
   \begin{tabular}{ccc}
  \hspace{-5mm}\includegraphics[width=0.95\textwidth]{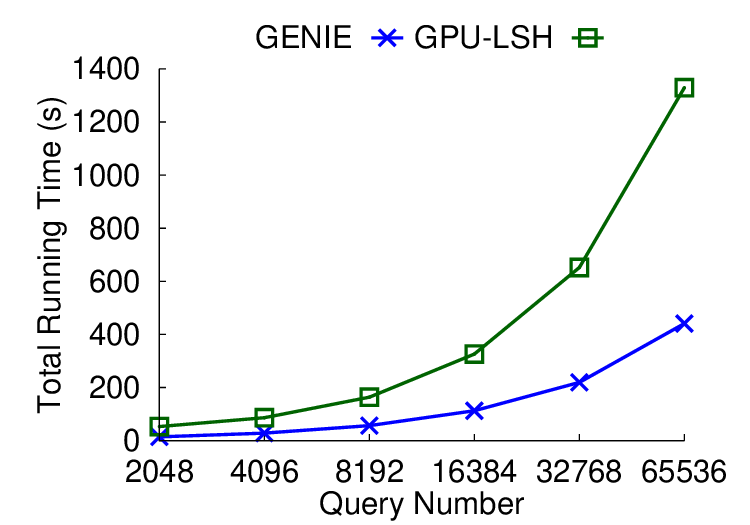}
    \end{tabular}
    }
 \vspace{-0mm}
\caption{   Running time with a large number of queries on SIFT data }
\label{fig:exp:total_time_vs_large_num_query_sift}
\end{minipage}
\vspace{-0mm}
\end{figure}

Figure \ref{fig:exp:total_time_vs_large_num_query_sift} shows the running time of GENIE and GPU-LSH for a larger number (up to 65536) of queries on SIFT data. Though GPU-LSH can support ten thousands of queries per batch, GENIE can also support such large number of queries with breaking query set into several small batches. 
With setting 1024 queries as a batch for GENIE, 
we can see that the time cost of GPU-LSH to process 65536 queries with one batch is 1329 seconds, while GENIE can process the same number of queries (with 64 batches) in 441 seconds. 

\subsubsection{Time profiling} \label{exp:timeprofile}


\begin{table}\small
\vspace{-0mm}
  \caption{\small Time profiling of different stages of GENIE for 1024 queries (the unit of time is $second$).  }\label{tab:profiling}
  \centering
\begin{tabular}{c|c|c|c|c|c|c}\hline
  \hline
  \multicolumn{2}{c|} {Stage} & OCR & SIFT& DBLP& Tweets & Adult\\\hline
  \multicolumn{2}{c|} {Index build}  & 81.39 & 47.73 & 147.34& 12.10 & 1.06\\\hline
  \multicolumn{2}{c|} {Index transfer} & 0.53  & 0.34 & 0.20& 0.088 & 0.011\\\hline
  \multirow{3}{*} {Query} &transfer & 0.015 & 0.018 & 0.0004& 0.0004 & 0.0004 \\\cline{2-7}
  &match & 2.60 & 7.04 & 0.85& 1.19 & 1.82 \\\cline{2-7}
  &select & 0.004 & 0.003 &0.11*  & 0.003 & 0.009 \\  \hline
  \hline
  \multicolumn{7}{c}{*This includes verification time which is the major cost.} \\ \hline
\end{tabular}
\vspace{-0mm}
\end{table}

Table \ref{tab:profiling} shows the time cost for different stages of GENIE. The ``Index-build'' represents the running time to build the inverted index on the CPU. This is an one-time cost, and we do not count it in the query time. The ``Index-transfer'' displays the time cost to transfer the inverted index from the CPU to the GPU. The rows of ``Query'' display the time for similarity search with 1024 queries per batch. The ``Query-transfer'' is the time cost to transfer queries and other information from the CPU to the GPU. The ``Query-select'' contains the time for selecting candidates from c-PQ and sending back the candidates to the CPU (For DBLP data, it also includes the time of verification).
The ``Query-match'' is the time cost for scanning inverted index which dominates the cost for similarity search. 
 This confirms our design choice of using GPU to accelerate this task.

\subsubsection{Experimental study on load balance} \label{sec:exp:lb}
\begin{figure}[htb]
\vspace{-0mm}
\centering
\begin{minipage}{.35\textwidth}
    \centerline{
   \begin{tabular}{ccc}
  \hspace{-5mm}\includegraphics[width=0.95\textwidth]{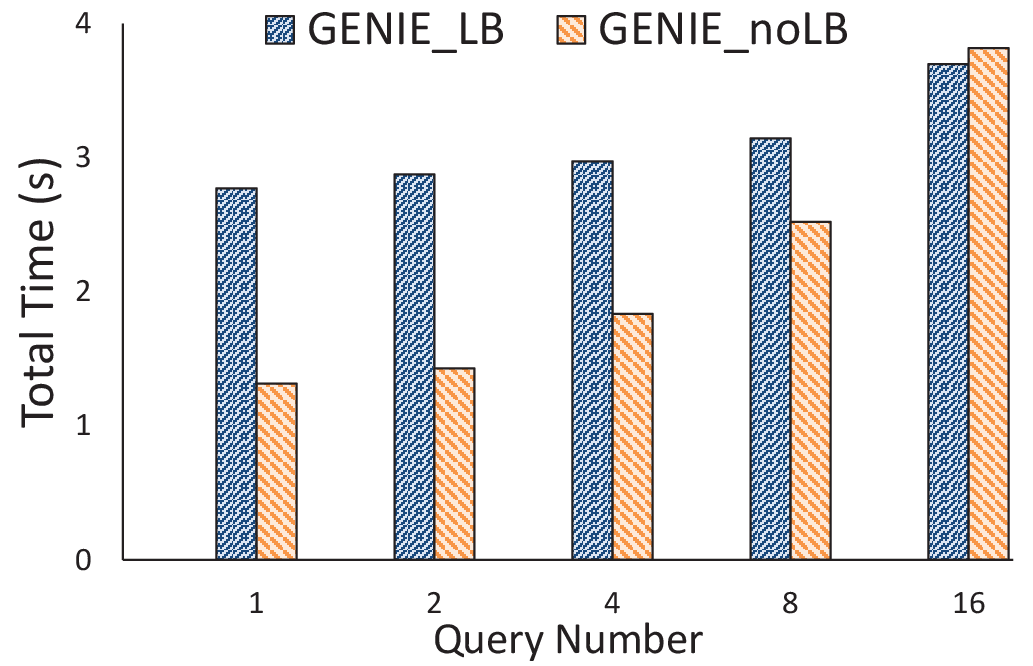}
    \end{tabular}
    }
 \vspace{-0mm}
\caption{   Load balance on Adult data (with 100M data points)}
\label{fig:exp:load_balance}
\end{minipage}
\vspace{-0mm}
\end{figure}

We study the effect of the load balance for GENIE on the Adult dataset which has long postings lists since some its attributes have only a few of categories (e.g. sex). We also duplicate the Adult dataset to 100M points to show the effect more clearly. In this experiment, we exert exact match for all attributes and return the best match candidates to the query. 
Figure \ref{fig:exp:load_balance} illustrates the running time of GENIE with and without load balance by varying the number of queries, where ``GENIE\_LB'' and ``GENIE\_noLB'' represent the running time with and without enabling load balance function respectively.
From Figure \ref{fig:exp:load_balance}, we can see that the load balance function can effectively allocate the workload to different blocks by breaking down the long lists. With increasing of the number of queries, the effect of the load balance is marginally decreased. The reason is that when the number of queries is larger, GENIE has already maximized the possibility for parallel processing by using one block for one postings list.
Besides, since the load balance requires some additional cost to maintain the index, the running time of GENIE with load balance is slightly higher than the one without load balance when the GPU is fully utilized.

\subsubsection{Searching on large data with multiple loadings}\label{sec:sift_multiple_loadings}
If the data set is too large to be processed with limited GPU memory, we adopt a multiple loading method (see Section \ref{sec:multiple_loadings}). Table \ref{tab:sift_multiple_loadings} shows the scalability of GENIE with different data sizes on SIFT\_LARGE dataset. In this experiment, we set the data part for each loading as 6M data points. By resorting to multiple loadings, GENIE can finish the query process for 1024 queries with 25.90 seconds on 36M SIFT data points. It also shows that GENIE with multiple loadings can scale up linearly with the number of data points. Since GPU-LSH cannot handle datasets with larger than 12M points, we estimate the running time of GPU-LSH on 24M and 36M data points with the same multiple loading method but without including index loading and result merge time. We can see that GPU-LSH has almost six times of running time of GENIE for the same dataset.

\begin{table}\small
  \caption{\small Running time of GENIE with multiple loadings on SIFT\_LARGE dataset for 1024 queries(unit: second)}\label{tab:sift_multiple_loadings}
  \centering
\begin{tabular}{c|c|c|c|c}  \hline
  \hline
  SIFT\_LARGE&6M&12M&24M&36M \\\hline
  GENIE & 4.32 & 8.62 & 17.26 & 25.90\\\hline
  GPU-LSH & 47.71	& 48.82 & (97.64)*& (146.46)*\\\hline
  CPU-LSH &2817 &5644 & 12333&20197 \\\hline
\hline
 \multicolumn{5}{c}{*It is estimated with multiple loading method.} \\ \hline
\end{tabular}
\vspace{-0mm}
\end{table}

\begin{table}\small
\vspace{-0mm}
  \caption{\small \small Extra running time cost of GENIE with multiple loadings with the same setting of Table \ref{tab:sift_multiple_loadings}(unit: second)}\label{tab:sift_multiple_loadings_extra}
  \centering
\begin{tabular}{c|c|c|c|c}  \hline
  \hline
  SIFT\_LARGE&6M&12M&24M&36M \\\hline
  Index transfer & 0.50 & 1.03 & 2.01 & 3.02\\\hline
  Result merge & 0 & 0.04&0.10 & 0.22\\\hline
  GENIE\_total & 4.32 & 8.62 & 17.26 & 25.90\\\hline
\hline
\end{tabular}
\vspace{-0mm}
\end{table}

GENIE with multiple loadings has two extra steps: 1) index loading: swapping index of each data part into the GPU memory and 2) result merging: merging the query results of each data part to obtain the final result. The running time cost of each extra step is shown in  Table \ref{tab:sift_multiple_loadings_extra}. We can see that the extra steps only take a small portion of the total time cost.


\subsubsection{Discussion}\label{exp:genie_discussion}

Here we give a brief discussion about the root causes that GENIE outperforms other methods.
It is not surprising that GENIE can outperform all the CPU-based algorithms like CPU-LSH, CPI-Idx and AppGram significantly. 

The key reason that GENIE can outperform GPU-LSH is due to the novel structure of c-PQ for candidate selection. The main bottleneck of GPU-LSH is to select top-k candidates from the candidate set generated by LSH, whose method essentially is to sort all candidates which is an expensive computation. Meanwhile, c-PQ can obtain the candidates by scanning the small Hash Table once. 

 GENIE outperforms GPU-SPQ with similar reasons. GPU-SPQ uses a k-selection algorithm on the GPU which requires multiple iterations to scan all candidates. Whereas GENIE only needs to scan the Hash Table once whose size (which is $O(k*AT)$) is much smaller than the candidate set of GPU-SPQ. 

\subsection{Effectiveness of c-PQ}\label{sec:exp:cpq}


\begin{figure*}[tb]
 \vspace{-0mm}
\centering
\begin{tabular}{ccccc}
  \hspace{-5mm}\includegraphics[width=0.2\textwidth]{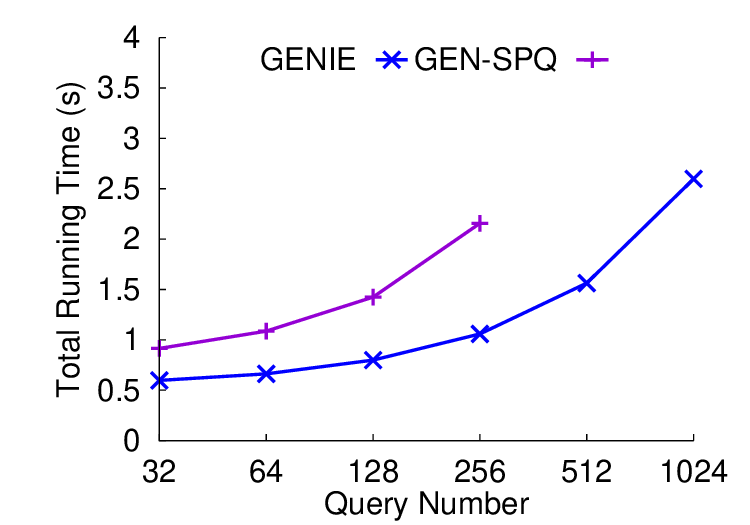} &
     \hspace{-5mm}\includegraphics[width=0.2\textwidth]{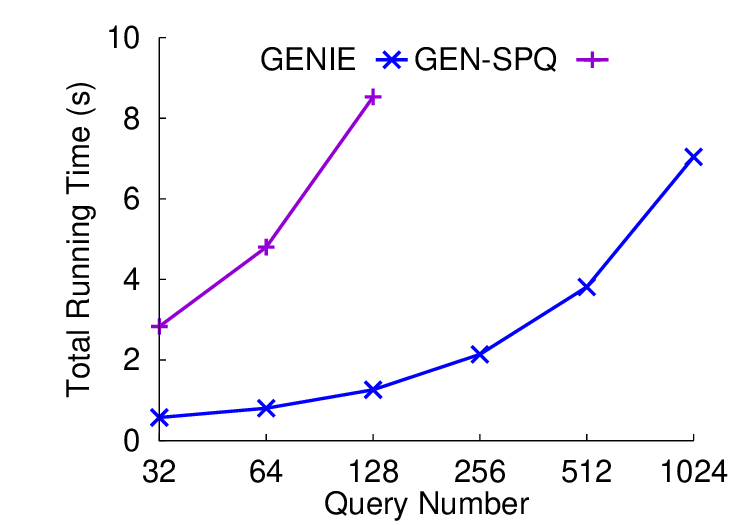}  & 
  \hspace{-5mm}\includegraphics[width=0.2\textwidth]{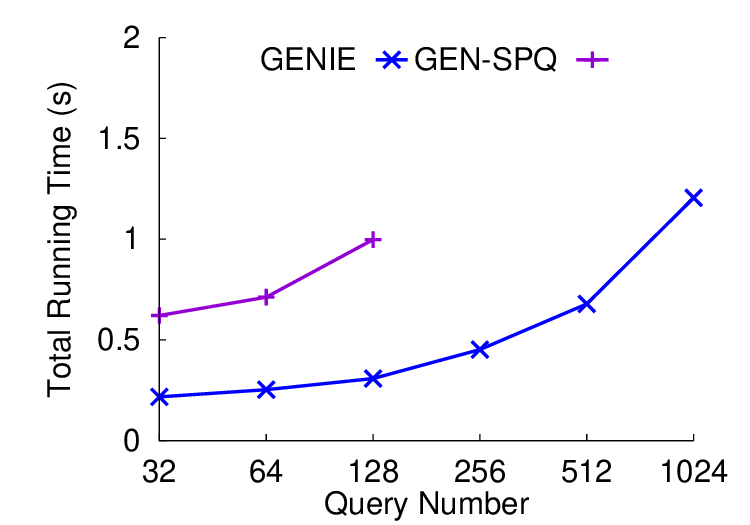}  &
      \hspace{-5mm}\includegraphics[width=0.2\textwidth]{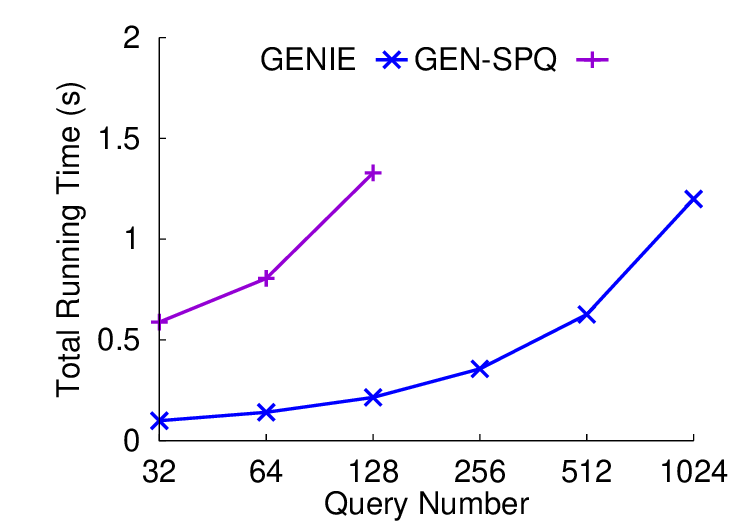}  &
     \hspace{-5mm}\includegraphics[width=0.2\textwidth]{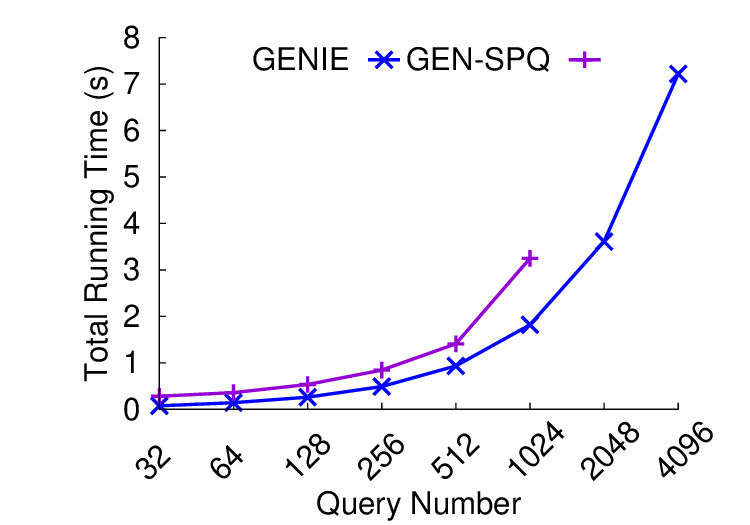} \\
 \vspace{-0mm}   (a) OCR &     (b) SIFT  &   (c) DBLP &   (d) Tweets & \hspace{-0.7cm}   (e) Adult   \\
\end{tabular}
\vspace{-0mm}
\caption{  The effectiveness of c-PQ.}
\label{fig:exp:genie_vs_genie_noch}
 \vspace{-0mm}
\end{figure*}

c-PQ can significantly reduce the memory requirement and the running time cost for GENIE. In Fig. \ref{fig:exp:genie_vs_genie_noch},  GEN-SPQ represents the running time of GENIE without c-PQ. 
We can see that, when the number of queries is the same, with the help of c-PQ the running time of GENIE decreases significantly since it avoids selecting candidates from a large Count Table. 


From Table \ref{tab:query_memory} we can see that GENIE reduces memory consumption per query to $1/5\thicksim1/10$ of the one of GEN-SPQ. To evaluate the memory consumption, with fixing the data size, we gradually increase the number of queries to find the maximum query number handled by our GPU, then we calculate the memory consumption per query by using 12 GB to divide the maximum query number.

\begin{table}\small
\vspace{-0mm}
  \caption{\small Memory consumption per query (unit: MB)}\label{tab:query_memory}
  \centering
\begin{tabular}{c|c|c|c|c|c}  \hline
  \hline
  dataset&OCR&SIFT&DBLP&Tweets&Adult \\\hline
  GENIE & 4.9 & 6.5 & 7.2 & 10.2& 1.4\\\hline
  GEN-SPQ & 41.0  & 49.1 &47.3 & 61.4& 8.8\\\hline
  \hline
\end{tabular}
\vspace{-0mm}
\end{table}


\subsection{Effectiveness of GENIE}
In this section, we evaluate the effectiveness of GENIE under the LSH scheme and the SA scheme. 
\subsubsection{ANN Search with GENIE}\label{sec:exp:ann}
Here we discuss the quality of the ANN search with GENIE as well as the parameter setting for GPU-LSH.
An evaluation metric for the ANN search is \emph{approximation ratio}, which is defined as how many times farther a reported neighbor is compared to the real nearest neighbor.
Formally, for a query point $q$, let $\{p_1,p_2,...,p_k\}$ be the ANN search results sorted in an ascending order of their $l_\mathfrak{p}$ normal distances to $q$. Let $\{p_1^*,p_2^*,...,p_k^*\}$ be the true $k$NNs sorted in an ascending order of their distances to $q$. Then the approximation ratio is formally defined as:
\begin{equation}
\vspace{-0mm}
\frac{1}{k}\sum_{i=1}^k \frac{ \| p_i-q\|_\mathfrak{p} }{\| p_i^* - q\|_\mathfrak{p}  }
\vspace{-0mm}
\end{equation}

When evaluating the running time, we set the parameters of GPU-LSH and GENIE to ensure that they have similar approximation ratio.

In the experiment evaluation (especially for running time) for ANN search  on the SIFT data set, we configure the parameters of GPU-LSH and GENIE to ensure that they have similar approximation ratio. For ANN search of GENIE, we set the number of hash functions as 237 which is determined by setting $\epsilon=\delta = 0.06$ (as discussed in Section \ref{sec:search:errbnd}). Another parameter for ANN search in high dimensional space is the bucket width (of Eqn. \ref{eqn:plsh}). According to the method discussed in the original paper of E2LSH \cite{datar2004locality}, we divide the whole hash domain into $67$ buckets. The setting of bucket width is a trade-off between time and accuracy: larger bucket width can improve the approximation ratio, but requires longer running time for similarity search.

For GPU-LSH, there are two important parameters: the number of hash functions per hash table and the number of hash tables. With fixing the number of hash tables, we find GPU-LSH has the minimal running time to achieve the same approximation ratio when the number of hash functions is 32. After fixing the number of hash functions as 32, we gradually increase the number of hash tables for GPU-LSH, until it can achieve approximation ratio similar to that of ANN search by GENIE ($k$ is fixed as 100). The number of hash tables is set as 700.


%

\begin{figure}[htb]
\vspace{-0mm}
\centering
\begin{minipage}{.4\textwidth}
    \centerline{
   \begin{tabular}{ccc}
  \hspace{-5mm}\includegraphics[width=0.95\textwidth]{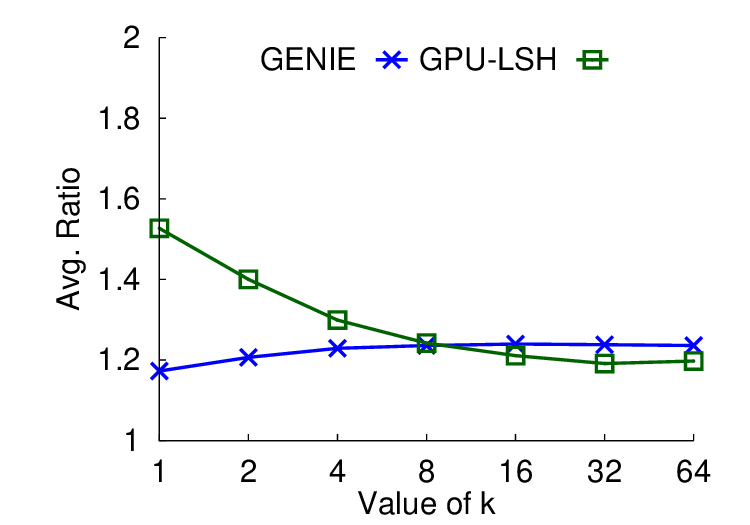}
    \end{tabular}
    }
 \vspace{-0mm}
\caption{   Approximation ratio v.s. value of k on SIFT data }
\label{fig:exp:ratio_vs_k}
\end{minipage}
\vspace{-0mm}
\end{figure}

%

Figure \ref{fig:exp:ratio_vs_k} shows the approximation ratio of GPU-LSH and GENIE, where GENIE has stable approximation ratio with varying $k$; whereas GPU-LSH has large approximation ratio when $k$ is small. 
The increase of approximation ratio of GPU-LSH with smaller $k$ is a common phenomenon which also appears in some pervious LSH methods like \cite{sun2014srs}. The reason is that these methods usually adopt some early-stop conditions, thus with larger $k$ they can access more points to improve the approximation ratio \cite{sun2014srs}. 

\begin{table}\small
 \vspace{-0mm}
  \caption{\small Prediction result of OCR data by 1NN }\label{tab:pred_ocr}
  \centering
\begin{tabular}{c|c|c|c|c}  \hline
  \hline
  method & precision & recall & F1-score & accuracy \\\hline
  GENIE & 0.8446 & 0.8348 &   0.8356 & 0.8374 \\\hline
  GPU-LSH &  0.7875 &  0.7730 & 0.7738 & 0.7783 \\\hline
  \hline
\end{tabular}
 \vspace{-0mm}
\end{table}

We use a similar method to determine the parameters for GPU-LSH and GENIE on the OCR dataset.
For ANN search in Laplacian kernel space by GENIE, except parameters $\epsilon$ and $\delta$ which determine the number of hash functions, another parameter is the kernel width $\sigma$ of the Laplacian kernel $k(x,y)=exp(-\|x-y\|_\mathfrak{1}/\sigma)$. We random sample 10K points from the dataset and use their mean of paired $l_1$ distance as the kernel width. This is a common method to determine the kernel width for kernel function introduced by Jaakkola et al.\cite{jaakkola1999using}.
GPU-LSH uses GPU's constant memory to store random vectors for LSH. Thus, the number of hash functions on OCR data cannot be larger than 8 otherwise the constant memory overflows. We use only 1M data points from the OCR dataset for GPU-LSH since it cannot work on a larger dataset. We increase the number of hash tables (with fixing the number of hash functions as 8) until it can achieve similar prediction performance as GENIE as reported in Table \ref{tab:pred_ocr}  where the number of hash tables for GPU-LSH is set as 100. Note that the prediction performance of GPU-LSH is slightly worse than the one of GENIE. It is possible to improve the performance of GPU-LSH by increasing the number of hash tables, which will dramatically increase the running time for queries. 

\subsubsection{Sequence similarity search with GENIE}\label{sec:exp:seq_search}
After finishing the search on GENIE, we can identify that some queries do not obtain real top-$k$ results (see Section \ref{sec:search:sequence}). 
Table \ref{tab:sequence_accu}  shows the percent of the queries obtaining correct top-1 search results
 with one round of the search process for 1024 queries. As we see from Table \ref{tab:sequence_accu}, with less than 10\% modification, GENIE can return correct results for almost all queries. Even with 40\% modification, GENIE can still return correct results for more than 95\% of queries. A typical application of such sequence similarity search is (typing) sequence error correction, where GENIE can return the most similar words (within minimum edit distance in a database)  for the 1024 queries with a latency time of 1 second (as shown in Figure \ref{fig:exp:total_time_vs_num_query} and Table \ref{tab:profiling}). Note that the one second is the whole latency time for 1024 queries including index transfer, query transfer, matching in GENIE and verification. AppGram may do this job with better accuracy, but it has much larger latency. A discussion on multiple round search and how to set  $\mathbf{K}$ is discussed in next section. 

\begin{table}\small
 \vspace{-0mm}
  \caption{\small Accuracy of top-1 search on DBLP dataset on GENIE (query length=40 and $\mathbf{K}=32$)}\label{tab:sequence_accu}
  \centering
\begin{tabular}{c|c|c|c|c}  \hline
  \hline
  Percent of modified & 0.1 & 0.2 & 0.3 & 0.4 \\\hline
  Accuracy & 1.0 & 0.999 &   0.995 & 0.954 \\\hline
Latency tiem (s)& 1.3 & 1.2 &  1.1 & 1.2 \\\hline
  \hline
\end{tabular}
 \vspace{-0mm}
\end{table}

\subsubsection{Sequence search with varying K}\label{sec:apx_sequence_accu_k}
{
\begin{table}\small
\vspace{-0mm}
  \caption{\small Accuracy and running time cost of GENIE for sequence search with varying $\mathbf{K}$ (query length=40).}\label{tab:sequence_accu_k}
  \centering
\begin{tabular}{c|c|c|c|c||c|c|c|c}  \hline
  \hline
  modified&0.1&0.2&0.3&0.4&0.1&0.2&0.3&0.4\\\hline
  $\mathbf{K}$&\multicolumn{4}{c||}{accuracy}&\multicolumn{4}{c}{time (seconds)} \\\hline
  8&0.999&0.998&0.979&0.908&1.3&0.98&1.0&0.96 \\\hline
 16&1.0&0.999&0.986&0.934&1.3&1.2&1.1&1.0 \\\hline
 32&1.0&0.999&0.995&0.954&1.3&1.2&1.1&1.2 \\\hline
  64&1.0&1.0&0.996&0.970&1.3&1.2&1.2&1.3 \\\hline
  128&1.0&1.0&0.997&0.975&1.4&1.3&1.4&1.6 \\\hline
  256&1.0&1.0&0.998&0.983&1.4&1.3&1.7&2.1 \\\hline
\end{tabular}
\vspace{-0mm}
\end{table}

}

With the help of Lemma \ref{lemma:top-$k$}, we can know whether the real $k$NN search result are returned. If users are keen on finding the true top-$k$ result, one possible solution is to repeat the search process by GENIE with larger $\mathbf{K}$. Table \ref{tab:sequence_accu_k} shows the accuracy and the running time of GENIE with varying the $\mathbf{K}$.  A method to find the real top-$k$ search result is to use a sequence of $\mathbf{K}=[8,16,32,64,128,...]$ as a multiple iteration method to find the results, but the running time cost is also high (which is accumulated with the time cost of different  $\mathbf{K}$). As we can see from Table \ref{tab:sequence_accu_k}, when the $\mathbf{K}$ is large enough (like $\mathbf{K}=64$), the increase of the accuracy becomes small. Our suggestion is to set a relatively large $\mathbf{K}$ which balances the running time and the accuracy. In our experiment, we set $\mathbf{K}=32$.

\section{Related work}\label{sec:related}


\subsection{Similarity search on different data}

Due to the ``curse of dimensionality'', spatial index methods 
provide little improvement over a linear scan algorithm when dimensionality is high. It is often unnecessary to find the exact nearest neighbour, leading to the development of LSH scheme for ANN search in high dimensional space \cite{indyk1998approximate}. 
We refer interested readers to a survey of LSH \cite{wang2014hashing}.%


The similarity between Sets, feature sketches and geometries is often known only implicitly, thus the computable kernel function is adopted for similarity search. 
To scale up similarity search on these data, the LSH-based ANN search in such kernel spaces has drawn considerable attention. 
Charikar \cite{charikar2002similarity} investigates several LSH families for kernelized similarity search. 
Wang et al. \cite{wang2014hashing} give a good survey about the LSH scheme on different data types. GENIE can support the similarity search in an arbitrary kernel space if it has an LSH scheme.


There is a wealth of literature concerning similarity search on complex structured data, and a large number of indexes have been devised. Many of them adopt the SA scheme \cite{aparicio2002whole,she2004shotgun} which splits the data objects into small sub-units and builds inverted index on these sub-units. Different data types are broken down into different types of sub-units. Examples include words for documents, 
n-grams for sequences \cite{wang2013efficient}, binary branches for trees \cite{yang2005similarity} and stars for graphs \cite{Graph:Yan:SIGMOD:2005}. 
Sometimes, a verification step is necessary to compute the real distance (e.g. edit distance) between the candidates and the query object \cite{wang2013efficient,Graph:Yan:SIGMOD:2005,yang2005similarity}.

\subsection{Parallelizing similarity search}
Parallelism can be adopted to improve the throughput for similarity search. 
There are also some proposed index structures on graphs and trees that can be parallelized \cite{tatikonda2010hashing,Graph:Yan:SIGMOD:2005}. However, indexes tailored to special data types cannot be easily extended to support other data types.

There are a few GPU-based methods for ANN search using LSH. Pan et al. \cite{pan2011fast,pan2012bi} propose a searching method on the GPU using a bi-level LSH algorithm, 
which specially designed for ANN search in the $l_\mathfrak{p}$ space. However GENIE can generally support LSH for ANN search under various similarity measures.

%
%

\subsection{Data structures and models on the GPU}
There are some works \cite{ding2009using,ao2011efficient} about inverted index on the GPU to design specialized algorithms for accelerating some important operations on search engines. An
inverted-like index on the GPU is also studied for continuous time series search \cite{zhou2015smiler}.
To the extent of our knowledge, there is no existing work on
inverted index framework for generic similarity search on the GPU.

Tree-based data structures are also investigated to utilize the parallel capability of the GPU.
Parallel accessing to the B-tree \cite{he2009relational} or R-tree \cite{luo2012parallel} index on the GPU memory is studied for efficiently handling query processing. 
Some GPU systems for key-value store are also studied \cite{zhang2015mega,hetherington2012characterizing}.

Researchers also propose general models to guide the query optimization on the GPU. For example, a GPU-based map-reduce framework is proposed to ease the development of data analysis tasks on the GPU \cite{he2008mars}.  
But none of them is specially designed for similarity search.

\subsection{Frequent item finding algorithm }
Some previous work related to Count Priority Queue (c-PQ) of GENIE is frequent item finding algorithm \cite{cormode2009finding}, which can be categorized as counter-based approach (like LossyCounting \cite{manku2002approximate} and SpaceSaving \cite{metwally2006integrated}) and sketch-based approach (like Count-Min \cite{cormode2005improved} and Count-Sketch \cite{charikar2004finding}). However, both approaches are approximation methods, whereas c-PQ can return the exact top-$k$ frequent count items.  Moreover, several frequent item finding algorithms (like SpaceSaving and Count-Min) require priority queue-like operations (e.g. finding the minimum from array), making them nontrivial be implemented on the GPU. 



\section{Conclusion}\label{sec:conclusion}


In this paper, we presented GENIE, a generic inverted index framework, which tries to reduce the programmer burden by providing a generic fashion for similarity search on the GPU for data types and similarity measures that can be modeled in the match-count model. Several techniques are devised to improve the parallelism and scaling out of the GPU, like $c$-PQ to reduce the time cost and the multiple loading method for handling large datasets. We also proved that GENIE can support $\tau$-ANN search for any similarity measure satisfying the LSH scheme, as well as similarity search on original data with the SA scheme.
In particular, we investigated how to use GENIE to support ANN search in kernel space and in high dimensional space, similarity search on sequence data and document data, and top-$k$ selection on relational data.
Extensive experiments on various datasets demonstrate the efficiency and effectiveness of GENIE.


\section*{Acknowledgments}
This research was carried out at the SeSaMe Centre. It is
supported by the Singapore NRF under its IRC@SG Funding Initiative and administered by the IDMPO.
The work by H. V. Jagadish was partially supported by the US National Science Foundation under Grants IIS-1250880 and IIS-1741022.

{\scriptsize
\vspace{-2mm}
\bibliographystyle{IEEEtran}
\bibliography{IEEEabrv,geniegpu}  
}
\newpage
\appendix

\section{k-selection on the GPU as Priority Queue}\label{sec:kselection}

We use a $k$-selection algorithm from a array to construct a priority queue-like data structure on the GPU, which is named as SPQ in the paper. To extract the top-$k$ object from an array, we modify a GPU-based bucket-selection algorithm \cite{alabi2012fast} for this purpose. Figure \ref{fig:bucketSelection} shows an example for such selection process.

\begin{figure}[htb]
\vspace{-0mm}
\centerline{
\includegraphics[width=0.35\textwidth]{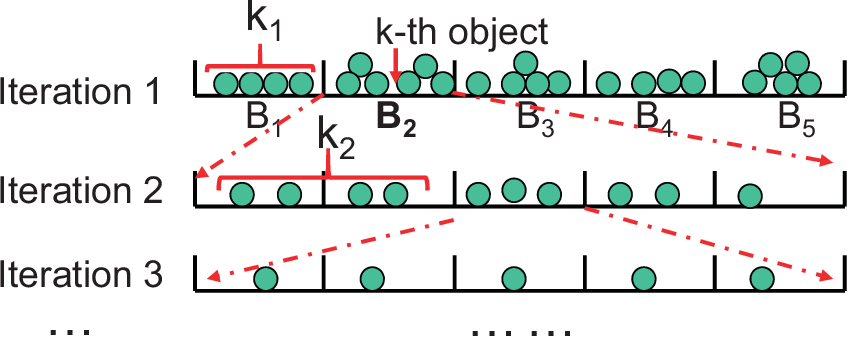}
}
\vspace{-0mm}
 \caption{Example for bucket k-selection. The iteration repeats until $k=k_1+k_2+...+k_t$.}
\label{fig:bucketSelection}
\vspace{-0mm}
\end{figure}

The algorithm has multiple iterations, and each iteration has three main steps. Step (1): we use a partition formulae $bucket_{id} = \lfloor(count - min)/(max - min) * bucket\_num)\rfloor$ to assign every objects into buckets. In Figure \ref{fig:bucketSelection}, all the objects in a hash table are assigned bucket $B_1$ to $B_5$. Step (2): We then check the bucket containing $k-th$ object. In Figure \ref{fig:bucketSelection}, $B_2$ is the selected bucket in Iteration 1. Step (3): we save all the objects before the selected bucket, and denote the number of saved objects as $k_i$ (e.g. $k_1$ in Iteration 1). Then we repeat Step (1)-(3) on the objects of the selected buckets until we find all the top-$k$ objects (i.e. $k=k_1+k_2+...+k_t$). From the algorithm, we can see that this method requires to explicitly store the object ids to assign them into buckets. In regard to multiple queries, we use one block to handle one hash table to support parallel selection in the GPU-based implementation. In our experiment, the algorithm usually finishes in two or three iterations.

\end{document}